\documentclass[11pt]{article}

\usepackage[T1]{fontenc}
\usepackage[utf8]{inputenc}
\usepackage{amsmath,amssymb,amsthm,amsfonts,mathtools}
\usepackage[margin=1in]{geometry} 
\usepackage{graphicx}

\theoremstyle{plain}
\newtheorem{theorem}{Theorem}[section]
\newtheorem{proposition}[theorem]{Proposition}
\newtheorem{lemma}[theorem]{Lemma}

\theoremstyle{definition}
\newtheorem{definition}[theorem]{Definition}
\newtheorem{remark}[theorem]{Remark}
\numberwithin{equation}{section}

\usepackage[pdfborder={0 0 0},bookmarksopen]{hyperref}
\hypersetup{%
  pdftitle = {Inventory Management for High-Frequency Trading with Imperfect Competition},
  pdfauthor = {Sebastian Herrmann, Johannes Muhle-Karbe, Dapeng Shang, Chen Yang},
  pdfpagemode = UseNone,
}

\newcommand{\wh}{\widehat}
\newcommand{\wt}{\widetilde}
\newcommand{\EE}{\mathbb{E}}
\newcommand{\RR}{\mathbb{R}}

\usepackage{color}

\usepackage{comment}

\begin{document}

\title{Inventory Management for High-Frequency Trading\\ with Imperfect Competition\footnote{The pertinent remarks of two anonymous reviewers are gratefully acknowledged.}}
\date{June 2, 2019}
\author{Sebastian Herrmann\thanks{Department of Mathematics, University of Michigan, email \texttt{sherrma@umich.edu}.}
\and
Johannes Muhle-Karbe\thanks{Department of Mathematics, Imperial College London, email \texttt{j.muhle-karbe@imperial.ac.uk}. Research partially supported by the CFM-Imperial Institute for Quantitative Finance. Parts of this paper were written while this author was visiting the Forschungsinstitut f\"ur Mathematik at ETH Z\"urich.}
\and
Dapeng Shang\thanks{Questrom School of Business, Boston University, email \texttt{dpshang@bu.edu}.}
\and
Chen Yang\thanks{Department of Mathematics, ETH Z\"urich, email \texttt{chen.yang@math.ethz.ch}. Partly supported by the Swiss National Foundation grant SNF $200020\_172815$.}
}\maketitle

\begin{abstract}
We study Nash equilibria for inventory-averse high-frequency traders (HFTs), who trade to exploit information about future price changes. For discrete trading rounds, the HFTs' optimal trading strategies and their equilibrium price impact are described by a system of nonlinear equations; explicit solutions obtain around the high-frequency limit. Unlike in the risk-neutral case, the optimal inventories become mean-reverting and vanish as the number of trading rounds becomes large. In contrast, the HFTs' risk-adjusted profits and the equilibrium price impact converge to their risk-neutral counterparts. Compared to a social-planner solution for cooperative HFTs, Nash competition leads to excess trading, so that marginal transaction taxes in fact decrease market liquidity.
\end{abstract}


\vspace{2mm}

\vspace{2mm}

\section{Introduction}

\emph{Information}, \emph{inventories}, and \emph{competition} are crucial elements of high-frequency trading. Indeed, many high-frequency strategies are based on gaining access to proprietary information and exploiting it before it becomes public knowledge. A typical example is ``latency arbitrage'', where high-frequency traders act based on price changes on other exchanges before these are incorporated into the consolidated ``national best bid-offer (NBBO) price''.

These strategies are not based on any longer-term view on the market. Accordingly, the main associated risk is the inventory that is built up along the way, which exposes the trader to adverse price moves. Thus, the natural tradeoff in this context is to exploit the available information as much as possible while simultaneously controlling the associated inventories. Crucially, HFTs do not make these decisions in isolation but in an environment where several competitors try to implement very similar strategies. 

In this paper, we study an equilibrium model with asymmetric information in the spirit of \cite{kyle.85,admati.pfleiderer.88} that allows us to analyze the interplay of these features in a tractable manner. We consider a market where risk-neutral, competitive dealers clear the orders of exogenous noise traders as well as several strategic HFTs. As in the latency-arbitrage trades mentioned above, these HFTs have access to future asset value changes one period before they become public knowledge. They in turn trade in Nash competition to exploit this additional information. Like in the classic literature on inventory costs (see, e.g., \cite{ho.stoll.81, madhavan.smidt.93}) and many more recent papers), their inventory is penalized. For tractability, this is modelled by a quadratic running cost as in \cite{muhlekarbe.webster.17,rosu.16}.

Since our model is based on a discrete informational advantage, we start from the equilibrium in a discrete-time setting and then study its convergence as the trading frequency increases to the high-frequency limit.\footnote{These convergence results provide a new justification for the direct study of models with infinitesimally short trading rounds as in \cite{rosu.16}.} In order to make the discrete-time model stationary, we postpone the terminal time indefinitely and show that a linear equilibrium exists for sufficiently frequent (but discrete) trading. Here, ``linear'' means that the dealers break even by adjusting the publicly known part of the asset value linearly for the net order flow (as in risk-neutral versions of the model, cf.~\cite{kyle.85,admati.pfleiderer.88} and many more recent studies) \emph{and} for the positions accumulated by the HFTs. Since these are not observable in practice, the dealer uses a linear forecast based on market observables. To obtain a consistent equilibrium, we in turn require that the HFTs have no incentive to deviate from the dealers' predictions. The HFTs' optimal strategies then also turn out to be linear in their positions as well as in their signals about future price changes. A key tractability feature of our equilibrium is that for an individual HFT's optimization problem given the equilibrium pricing rule, the other HFTs' predicted inventories are irrelevant; cf.~Remark~\ref{rem:tractability}. 

Our equilibrium is determined by a system of nonlinear algebraic equations. As the discretization parameter tends to zero, all relevant quantities admit asymptotic expansions around their high-frequency limit. Irrespective of the HFTs' inventory aversions, several aspects of this limit correspond to the risk-neutral version of the model studied by \cite{admati.pfleiderer.88}. Indeed, market depth, the HFTs' exploitation of future signals, and their optimal performance all converge to their risk-neutral counterparts. In contrast, the corresponding inventories are reduced to zero as the trading frequency increases. Highly frequent inventory management thereby allows the HFTs to achieve almost the same performance as without risk penalties.

The impact of inventory aversion on market liquidity and welfare becomes visible at the next-to-leading order in our expansions. However, we illustrate through numerical examples that the magnitude of these welfare and liquidity effects is small on the very short time scales relevant for high-frequency trading. In contrast, Nash competition between several HFTs plays a crucial role in our equilibrium like in the risk-neutral equilibrium in \cite{admati.pfleiderer.88}. Indeed, each HFT only internalizes the negative effects of her price impact on herself but not on others. Therefore, each of them trades too aggressively compared to the efficient allocation that would be achieved by coordinating through a central planner. As a consequence, increased competition between HFTs improves market liquidity.

In contrast, transaction taxes have the opposite effect. Indeed, as first observed in a one-period model with risk-neutral HFTs by \cite{subrahmanyam.98}, a (quadratic) transaction tax forces the HFTs to scale back their trading, thereby moving them closer to their collusive optimum.  As succinctly summarized by \cite{subrahmanyam.98}, ``this causes the transaction tax to have a perverse effect: it reduces market liquidity (and increases the adverse price impact faced by informationless traders), but it also reduces informed trader profits''. This is an important mechanism to keep in mind when discussing transaction taxes as a tool to improve market quality by curbing high-frequency trading. To become socially optimal, these negative effects of taxation have to be outweighed by reducing costs for information acquisition~\cite{subrahmanyam.98} or (over-)investment in trading technologies~\cite{biais.al.15} for example. Incorporating such features into the present dynamic equilibrium model is an intriguing but challenging direction for further research.

The present paper builds on the model studied in \cite{muhlekarbe.webster.17}, extending it in two crucial directions. First, we study imperfect Nash competition between several HFTs with possibly heterogeneous risk-aversion parameters,\footnote{Other types of heterogeneity have been considered in the literature. For example, \cite{hirshleifer.al.94} compare two groups of risk-averse HFTs that acquire private information at different rounds in a two-period model. \cite{foster.vis.96} consider heterogeneity in the initial private signals; see also \cite{back.al.00} for a continuous-time version of this model.} rather than focusing on a single informed agent. Second, we solve for a full equilibrium, where the execution price allows the dealers to break even in each trading round rather than only on average over the whole trading interval. We show that in the high-frequency limit, most predictions of the model are robust. This provides some justification for the use of simpler pricing rules as in \cite{muhlekarbe.webster.17} or \cite[Section~4.1]{rosu.16}.

The HFTs in our model repeatedly receive short-lived informational advantages in the sense that all their private signals are revealed after only one trading round as in \cite{admati.pfleiderer.88}. This complements a large body of literature dealing with various kinds of long-lived informational advantages \cite{kyle.85,holden_long-lived_1992, holden_risk_1994, foster.vis.96, back.pedersen.98, back.al.00, chau.vayanos.08, caldentey.stacchetti.10, foucault.al.13}, mostly in risk-neutral settings.
Closer to our work is Ro\c{s}u \cite{rosu.16}, who studies a model of long-lived informational advantage, but also considers a quadratic inventory penalty for informed traders.
To deal with the resulting more involved filtering problems, \cite{rosu.16} directly works in a model with ``infinitesimally short'' times between trades and optimizes in several parametric classes of trading strategies. In contrast, we start from a fully discrete model, solve for the corresponding equilibrium, and then study its convergence to a high-frequency limit.\footnote{In particular, the strategies of our HFTs are determined by an optimal tradeoff between expected profits and inventory costs at any trading frequency. In contrast, the inventory cost vanishes in the infinitesimal model of \cite{rosu.16}, so that both the optimal weight placed on new signals and the rate of inventory management are determined by maximizing expected profits from trades with slower agents.} With our information structure, linear autoregressive strategies indeed turn out to be optimal. However, whereas HFTs in the model of \cite{rosu.16} ``no longer speculate on the long term value, but just pass their inventory to slower traders'', our model predicts that they initially exploit their trading signals in a very similar manner as their risk-neutral counterparts but quickly unwind their positions after their informational edge has evaporated.

This article is organized as follows. In Section~\ref{sec:market}, we introduce the basic ingredients of our market with asymmetric information. Subsequently, in Section~\ref{sec:monopoly}, we set the stage by discussing the benchmark case of a single HFT who monopolizes access to the additional information. (By a change of variables, this can also be interpreted as a collusive equilibrium, where several homogeneous HFTs coordinate by letting a social planner maximize their aggregate welfare.) The full version of our model with imperfect competition between several, possibly inhomogeneous HFTs is in turn studied in Section~\ref{sec:competition}. For better readability, all proofs are delegated to the appendix.

\section{Market}
\label{sec:market}

We consider a market for a risky asset. At each time 
$$t_n=n \Delta t, \quad \Delta t>0, \quad n=1,2,\ldots,$$
a new increment $\Delta S_n$ of its fundamental value is revealed. This can be interpreted as the change of the National Best Bid Offer (NBBO) price aggregated across exchanges. The value increments are independent and identically normally distributed with mean zero and variance $\sigma_S^2\Delta t>0$.\footnote{That is, the high-frequency limit of the fundamental value is a Brownian motion $W^S$ with volatility $\sigma_S$.} The risky asset is traded at times $t_n$, $n=1,2,\ldots$, by three types of market participants:
\begin{enumerate}
\item \emph{Noise traders}, who trade according to their own objectives (e.g., long-term investment or hedging) and have no private information about the next value increment. Their aggregate order size at time $t_n$ is modeled by an exogenous random variable $\Delta K_n$, which is normally distributed with mean zero and variance $\sigma_K^2\Delta t>0$ and independent from the fundamental value of the risky asset.\footnote{That is, the high-frequency limit of the noise trades is a Brownian motion $W^K$ independent of the one driving the fundamental value.}
\item \emph{High-frequency traders (HFTs)}, with initial inventory $L_0 \in \RR$, who already observe the increment $\Delta S_{n}$ of the asset value before trading at time $t_n$ and utilize it to choose their trade $\Delta L_n$ at time $t_n$ to maximize their discounted expected profits penalized for the corresponding squared inventories.\footnote{The more realistic case of \emph{noisy} signals is a challenging direction for future research; see Remark~\ref{rem:noisy signals}.}

\item Risk-neutral, competitive \emph{dealers}, who set an execution price $P_n$ at time $t_n$ at which they can clear the market while achieving zero expected profits conditional on each realization of the net order flow $\Delta Y_n = \Delta K_n+\Delta L_n$ and the past value increments $\lbrace \Delta S_m,m<n \rbrace$. This means that the dealers incorporate information about the fundamental value that has already become public knowledge. Conditioning on the current order flow means that prices are volume-dependent, akin to quotes in a limit-order book.
\end{enumerate}

\section{Equilibrium with Monopolistic Insider}
\label{sec:monopoly}

We first consider the case of a single HFT, who monopolizes the access to the information about the next-period value increment. 

\subsection{Dealers' Pricing Rule}
\label{sec:monopoly:dealer}

The dealers' information set just before the $n$-th trading round consists of the past value increments $\lbrace \Delta S_m, m<n \rbrace$ and the past and current order flow $\lbrace \Delta Y_m, m\leq n\rbrace$, and we denote the corresponding conditional expectation operator by $\wt \EE_n[\cdot]$. They quote an execution price $P_n$ that allows them to break even on average in each trading round:
\begin{align*}
P_n 
&= \wt \EE_n[S_n] =  S_{n-1} + \wt\EE_n[\Delta S_n],\quad n=1,2,\ldots.
\end{align*}
We focus on \emph{linear pricing rules} of the form
\begin{align}
\label{eqn:monopoly:price}
P_n
&= S_{n-1} + \lambda\Delta Y_n + \mu M_{n-1}.
\end{align}
We do \emph{not} assume that the dealers can observe the evolution of the HFT's inventories. Therefore, instead of the actual position, the dealers' pricing rule depends on the \emph{inventory prediction process} $M = (M_n)_{n\geq 0}$ given by $M_0 = L_0$ and
\begin{equation}
\label{eqn:monopoly:inventory prediction process}
\Delta M_n
= \beta \Delta S_n - \varphi M_{n-1}, \quad \mbox{for some $\beta > 0$ and $\varphi \in (0,1]$.}
\end{equation}
The execution price $P_n$ in turn adjusts the fundamental value $S_{n-1}$ that has already been revealed linearly for the dealers' prediction $M_{n-1}$ of the HFT's inventory and for the new order flow $\Delta Y_n$. To obtain a consistent equilibrium, we require that the HFT has no incentive to deviate from the dealers' inventory prediction, so that this estimate indeed is accurate, cf.~Definition~\ref{def:monopoly:equilibrium} below.

\begin{remark}
All four coefficients $\mu,\lambda,\beta,\varphi$ are expected to be positive in equilibrium:
\begin{itemize}
\item The order-flow sensitivity $\lambda$ is similar to ``Kyle's lambda'', in that it allows the dealers to charge higher prices if large order flows indicate that the HFT wants to buy to exploit increases in the asset value.
\item The inventory sensitivity $\mu$ adjusts the execution price for the rebalancing trades the HFT implements to manage her inventory. For instance, a very large inventory makes it attractive for the HFT to sell, so that the same positive order flow is likely to correspond to an even larger value increase that needs to be offset by a higher execution price.
\item The signal sensitivity $\beta$ describes how strongly the HFT is predicted to react to their information about value changes.
\item The inventory sensitivity $\varphi$ determines how fast the HFT is predicted to reduce large (positive or negative) inventories.
\end{itemize}
\end{remark}

\begin{remark}
We assume that the initial inventory prediction is correct, i.e., $M_0 = L_0$. In reality, the dealers do not know the actual initial inventory of the HFT. However, as HFTs are very reluctant to hold inventories overnight, a zero inventory $L_0 =0$ at the start of a trading day is a reasonable assumption in our context. Moreover, as we discuss in Remark~\ref{rem:stability}, if the initial inventory prediction is incorrect, it is optimal for the HFT to gradually reduce that discrepancy.

The case where the initial inventory $L_0$ is random from the point of view of dealers leads to similar difficulties as the case of noisy signals (see Remark~\ref{rem:noisy signals}).
\end{remark}

\subsection{HFT's Optimization}
\label{sec:monopoly:HFT}

The HFT has a discount rate $\rho \Delta t \in (0,1)$ and, as in \cite{rosu.16,muhlekarbe.webster.17}, a holding cost $\gamma\Delta t/2>0$ levied on her squared inventory in the risky asset. (This means that $\rho$ and $\gamma$ are the discount rate and inventory cost per unit time, respectively.) The HFT's information set just before the $n$-th trading round consists of the past and current value increments $\lbrace \Delta S_m, m\leq n\rbrace$, and we denote the corresponding conditional expectation by $\EE_n[\cdot]$.

We next derive the HFT's goal functional. We postpone the terminal time to infinity to obtain a stationary objective. To wit, the HFT optimizes the discounted sum of her expected one-period wealth changes, penalized for holding large inventories, similarly to \cite{garleanu.pedersen.13}. Consider trading round $n$. Prior to trading, the HFT holds $L_{n-1}$ shares of the risky asset (and a cash position) and submits an order of size $\Delta L_n$. This order is executed at the dealers' execution price $P_n$. Since the market price is based on the dealers' inferior information, we assume that the HFT uses her own, more accurate forecast $S_n$ of the fundamental value to evaluate her risky position. Thus, the HFT's wealth change due to trading in round $n$ is
\begin{align*}
L_n S_n - L_{n-1}S_{n-1} - P_n\Delta L_n 
&= L_{n-1}\Delta S_n + (S_n - P_n) \Delta L_n;
\end{align*}
here, $L_{n-1}S_{n-1}$ and $L_nS_n$ are the HFT's valuations of the risky position before and after trading in round $n$, respectively, and $P_n\Delta L_n$ is the change in the cash position due to the trade $\Delta L_n$. Since $\Delta S_n$ and $L_{n-1}$ are independent and $\Delta S_n$ has mean $0$, the wealth change $L_{n-1}\Delta S_n $ corresponding to the shares already held before trading round $n$ has expectation zero, and it suffices to consider the wealth change $(S_n-P_n)\Delta L_n$ corresponding to the purchase or sale of $\Delta L_n$ shares in round $n$.



Suppose that the dealers have chosen a pricing rule $(\lambda,\mu,\beta,\varphi)$ with corresponding inventory prediction process $M$. Then, purchasing $\Delta L_n$ new shares at time $t_n$ costs 
$$P_n \Delta L_n = (S_{n-1}+ \lambda \Delta Y_n +\mu M_{n-1})\Delta L_n.$$
Comparing this to the HFT's valuation $S_n = S_{n-1}+\Delta S_n$ and taking into account that the noise trades are independent with mean zero, it follows that the expected change of the HFT's wealth due to the new trade is $(\Delta S_{n}-\lambda\Delta L_n-\mu M_{n-1})\Delta L_n$. 

Now, complement this with the inventory penalty $\frac{\gamma\Delta t}{2} (L_{n-1}+\Delta L_n)^2$ for trading round $n$ and discount the expected wealth change and inventory penalty at time $t_n$ with the simply compounded discount factor $(1-\rho\Delta t)^{n}$. For an infinite planning horizon, this leads to the following stationary goal functional:
\begin{align}
\label{eqn:monopoly:objective}
\EE_1\bigg[ \sum_{n=1}^\infty (1-\rho\Delta t)^n \Big\lbrace\big( \Delta S_n - \lambda \Delta L_n - \mu M_{n-1}\big) \Delta L_n -\frac{\gamma\Delta t}{2}(L_{n-1}+\Delta L_n)^2 \Big\rbrace \bigg] \to \max_{\Delta L \in \mathcal{A}}!
\end{align}
Here, 
\begin{align*}
\mathcal{A}
&=\Big\{(\Delta L_n)_{n\geq 1}: \EE[L_n^2] \text{~is bounded in $n$} \Big\}
\end{align*}
denotes the set of \emph{admissible strategies}. We note that $\Delta M \in \mathcal{A}$ by Lemma~\ref{lem:admissible:general}. Using this, one can show that the expectation in \eqref{eqn:monopoly:objective} is well defined and finite for all $\Delta L \in \mathcal{A}$.

\subsection{Definition of Equilibrium}
\label{sec:monopoly:equilibrium}

We can now formalize our notion of equilibrium. As already outlined above, the key consistency condition is that the dealers' inventory predictions indeed come true, because the HFT has no incentive to deviate from them:

\begin{definition}
\label{def:monopoly:equilibrium}
Let $(\lambda,\mu,\beta,\varphi)$ be a pricing rule with corresponding inventory prediction process $M$ as in~\eqref{eqn:monopoly:inventory prediction process}. We say that $(\lambda,\mu,\beta,\varphi)$ forms a \emph{(linear) equilibrium} if:
\begin{enumerate}
\item Given the pricing rule $(\lambda,\mu,\beta,\varphi)$, the dealers' inventory prediction $\Delta M$ is an admissible strategy and optimal for the HFT's goal functional \eqref{eqn:monopoly:objective}.
\item Given that the HFT uses the strategy $\Delta M$, the dealers' conditional expected profits in each trading round are zero:
\begin{align}
\label{eqn:def:monopoly:equilibrium}
\wt\EE_n[\Delta S_n]
&= \lambda \Delta Y_n + \mu M_{n-1},\quad n\geq 1.
\end{align}
\end{enumerate}
\end{definition}

A couple of remarks are in order.

\begin{remark}
The coefficients $(\lambda, \mu, \beta, \varphi)$ of a linear equilibrium are time-independent. This is due to two facts:
\begin{enumerate}
\item The HFT's optimization problem has an infinite time horizon. If the HFTs optimization problem had a finite time horizon, optimal strategies (and then also pricing rules) would become time-dependent.

\item The HFT's informational advantage is fully reset after each trading round: the old increment is revealed to the public, the HFT receives a new signal which is uncorrelated with previous observations, and past order flow observations become irrelevant. If, from the dealers' point of view, there was unrevealed randomness in the HFT's inventory (e.g., due to a random initial inventory or noisy signals), then the dealers would use past order flow observations to compute expected value increments and the pricing rule would likely become time-dependent.
\end{enumerate} 
\end{remark}

\begin{remark}
\label{rem:noisy signals}
Recall that in our model, the HFT observes $\Delta S_n$ before trading at $t_n$. It is more realistic to assume that the HFT only receives a noisy (correlated) observation $\wt{\Delta S}_n$ of $\Delta S_n$. In the risk-neutral version \cite{admati.pfleiderer.88} of the model, noisy signals can be incorporated without much difficulty. However, with inventory aversion, challenging problems appear.

Consider the dealers' situation before the $n$-trading round. In view of the zero-profit condition \eqref{eqn:def:monopoly:equilibrium}, the dealers need to compute the conditional expectation $\wt\EE_n[\Delta S_n]$ based on the past and current order flows $\lbrace \Delta Y_m : m \leq n \rbrace$ and past value increments $\lbrace \Delta S_m : m < n\rbrace$. 

In our equilibrium, the dealers can perfectly predict the HFT's inventory in terms of $\lbrace \Delta S_m : m < n\rbrace$. Hence, the dealers can predict the HFT's trade in round $n$ as a function of $\Delta S_n$. Since past order flows $Y_1,\ldots,Y_{n-1}$ are uncorrelated with the next value increment $\Delta S_n$, only $\Delta Y_n$ is relevant for computing the conditional expectation $\wt\EE_n[\Delta S_n]$.

If the HFT receives noisy signals, the situation is different. From the dealers' point of view, the HFT's noisy signals lead to noise in the HFT's inventory (since it is unrealistic to assume that the noisy signals are being revealed to the public after each trading round). Thus, at best, the dealers can predict the HFT's trade as a function of $\Delta S_n$ and $L_{n-1}$. Hence, $Y_n$ is correlated with both $\Delta S_n$ and $L_{n-1}$. But $L_{n-1}$ also depends on all previous noisy signals. Consequently, one expects that the conditional expectation $\wt\EE_n[\Delta S_n]$ depends not only on the current order flow $Y_n$ but also on all past order flows $Y_1,\ldots,Y_{n-1}$. This leads to considerably more complex computations.
\end{remark}

\subsection{Existence and Asymptotics}
\label{sec:monopoly:verification}

If the trading frequency is sufficiently high, a linear equilibrium exists and is unique in the proposed class. All relevant quantities can be expressed in terms of the solution of a quartic equation. In the high-frequency limit, this leads to closed-form approximations by means of the implicit function theorem. The proof of these results is a special case of the more general Theorem~\ref{thm:competition} below, which is established in Appendix~\ref{sec:proofs}.

\begin{theorem}~
\label{thm:monopoly}
\begin{enumerate}
\item For sufficiently small $\Delta t \geq 0$, the quartic equation
\begin{align}
\label{eqn:thm:monopoly:quartic}
0
&= \beta^4 (1-\rho\Delta t)-(2-\rho\Delta t +\beta  \gamma  \Delta t)\left(\frac{\sigma_K}{\sigma_S}\right)^2\beta ^2 +\left(\frac{\sigma_K}{\sigma_S}\right)^4 (1-\beta  \gamma  \Delta t)
\end{align}
has a unique solution $\beta \in (0,\frac{\sigma_K}{\sigma_S}]$. This solution has the following asymptotics as $\Delta t \to 0$:
\begin{align*}
\beta
&= \frac{\sigma_K}{\sigma_S}-\left(\frac{\gamma\sigma_K^3}{2\sigma_S^3}\right)^{1/2}\sqrt{\Delta t} + O(\Delta t).
\end{align*}

\item Define, for sufficiently small $\Delta t \geq 0$,\footnote{Note that $\lambda\beta < 1$, so that $\varphi$ is well defined.}
\begin{align}
\label{eqn:thm:monopoly:lambda}
\lambda
&= \frac{\beta \sigma_S^2}{\sigma_K^2+\beta^2\sigma_S^2}
= \frac{\sigma_S}{2\sigma_K}-\frac{\gamma}{8}\Delta t+O((\Delta t)^{3/2}),\\
\label{eqn:thm:monopoly:phi}
\varphi
&= 1- \frac{\lambda\beta}{1-\lambda\beta}
= \left(\frac{2\gamma\sigma_K}{\sigma_S}\right)^{1/2}\sqrt{\Delta t}+O(\Delta t),\\
\label{eqn:thm:monopoly:mu}
\mu
&= \lambda \varphi
= \left(\frac{\gamma\sigma_S}{2\sigma_K}\right)^{1/2}\sqrt{\Delta t}+O(\Delta t).
\end{align}
Then, for sufficiently small $\Delta t>0$, the pricing rule $(\lambda,\mu,\beta,\varphi)$ forms a linear equilibrium. In particular, the strategy
\begin{align*}
\Delta L_n
&= \Delta M_n
= \beta\Delta S_n - \varphi M_{n-1}
\end{align*}
is optimal for the HFT.\footnote{With this choice of $\Delta L_n$ and because $M_0 = L_0$ by assumption, the dealers' inventory prediction in the equilibrium is correct at all times, i.e., $M_n = L_n$ for all $n$. In particular, the HFT's optimal strategy could equivalently be formulated as $\Delta L_n = \beta \Delta S_n - \varphi L_{n-1}$.}

\item The optimal performance of the HFT at time $0$ is 
\begin{align*}
-\frac{A}{2}L^2_0 + \frac{B}{2}\Delta S_1^2 -C L_0\Delta S_1 + D,
\end{align*}
where
\begin{align*}
A
&= \frac{(1-\rho\Delta t)(1-\varphi)^2}{1-(1-\rho\Delta t)(1-\varphi)^2}\gamma \Delta t
= \frac{\sqrt{2}}{4} \gamma^{1/2}\left(\frac{\sigma_S}{\sigma_K}\right)^{1/2}\sqrt{\Delta t} + O(\Delta t),\\
B
&= (1-\rho\Delta t)\beta  (2(1-\lambda\beta)-\beta(A+ \gamma  \Delta t))
= \frac{\sigma_K}{\sigma_S} - \frac{\sqrt{2}}{4} \gamma^{1/2} \left(\frac{\sigma_K}{\sigma_S}\right)^{3/2}\sqrt{\Delta t} + O(\Delta t),\\
C
&=(1-\rho\Delta t)( \beta  (1-\varphi) (A+\gamma  \Delta t)+\varphi (1- \lambda\beta))
= \frac{3\sqrt{2}}{4} \gamma^{1/2} \left(\frac{\sigma_K}{\sigma_S}\right)^{1/2} \sqrt{\Delta t} + O(\Delta t),\\
D
&= \frac{(1-\rho\Delta t)B \sigma_S^2}{2\rho}
= \frac{\sigma_S\sigma_K}{2\rho} -\frac{\sqrt{2}}{8}\gamma^{1/2} \frac{\sigma_S^{1/2}\sigma_K^{3/2}}{\rho}\sqrt{\Delta t} + O(\Delta t).
\end{align*}
\end{enumerate}
\end{theorem}

\begin{remark}
The linear equilibrium from Theorem~\ref{thm:monopoly} is unique for sufficiently small $\Delta t > 0$. We outline the main steps of the argument.

Let $(\lambda,\mu,\beta,\varphi)$ be a linear equilibrium. The dealers' zero profit condition necessitates that $\lambda$ and $\mu$ are related to $\beta$ and $\varphi$ by \eqref{eqn:thm:monopoly:lambda} and \eqref{eqn:thm:monopoly:mu}, respectively. Next, one can use either dynamic programming or perturbation arguments to analyze the optimality of the strategy $(\Delta M_n)_{n\geq1}$ for the HFT, i.e., that the HFT does not have any incentive to deviate from the dealers' prediction. This yields that $\varphi$ is given in terms of $\lambda$ and $\beta$ by \eqref{eqn:thm:monopoly:phi} and that $\beta$ satisfies the quartic equation \eqref{eqn:thm:monopoly:quartic}. It thus remains to analyze which solutions to the quartic equation produce an equilibrium.

It turns out that the quartic equation \eqref{eqn:thm:monopoly:quartic} has precisely two distinct real roots. The first root lies in  $(0,\frac{\sigma_K}{\sigma_S})$ and leads to the equilibrium of Theorem~\ref{thm:monopoly}. The second root lies in $(\frac{\sigma_K}{\sigma_S},\infty)$. However, $\beta > \frac{\sigma_K}{\sigma_S}$ implies that the corresponding mean-reversion parameter $\varphi$ is negative. Hence, by Lemma~\ref{lem:admissible:general}, the corresponding strategy $(\Delta M_n)_{n\geq 1}$ is not admissible. Moreover, this strategy would also not be optimal in any larger admissibility class. Indeed, the performance of the strategy $(\Delta M_n)_{n\geq 1}$ can be computed explicitly in terms of $\beta$, $\varphi$, and $\lambda$ using the explicit representation for $M_n$ from the proof of Lemma~\ref{lem:admissible:general}. Then, it turns out that for $\varphi < 0$ and sufficiently small $\Delta t>0$, the inventory penalty dominates the expected profits and leads to the performance $-\infty$.
\end{remark}

\subsection{Comparative Statics}

We now discuss the optimal trading strategy, pricing rule, and optimal HFT performance for the equilibrium of Theorem~\ref{thm:monopoly}. The high-frequency limit $\Delta t \to 0$ of the equilibrium coincides with its risk-neutral counterpart~\cite{admati.pfleiderer.88}; the effects of inventory aversion only become visible in the leading-order correction terms for $\Delta t>0$. Indeed, in this limit, the optimal trading strategy $\Delta L_n = \frac{\sigma_K}{\sigma_S} \Delta S_n$ and the linear pricing rule $P_n = S_{n-1} +\frac{\sigma_S}{2\sigma_K} \Delta Y_n$ coincide with the risk-neutral equilibrium of~\cite{admati.pfleiderer.88}.

For the rest of this section, we turn to the case of sufficiently small, but positive $\Delta t$, where the impact of the inventory aversion becomes visible. For the numerical illustrations in Figures~\ref{fig:monopoly:strategy}--\ref{fig:monopoly:pricing rule}, we use the volatilities $\sigma_S=\sigma_K=1$ as in \cite{subrahmanyam.98}, the inventory aversion parameter $\gamma = 1$, which is of the same order of magnitude as the parameter values used in \cite{rosu.16}, and the discount rate $\rho=5\%$; note, however, that  as in the high-frequency limits, almost all results are virtually independent of $\rho$ in any case. (The only exception is the optimal performance $D$, for which both the limit and the leading-order correction are inversely proportional to the discount rate.)

\paragraph{Trading strategy.} 
The HFT's equilibrium position $(L_n)_{n\geq0}$ follows an autoregressive process of order one, where the mean-reversion speed is governed by the parameter $\varphi$ and the innovations are given by $\beta$ times the value increments $\Delta S_n$. The mean-reversion component ensures that the HFT's position does not become too large (positive or negative). Note that this inventory management term does \emph{not} scale with $\Delta t$ as for a discretized Ornstein--Uhlenbeck process. Instead, it scales with $\sqrt{\Delta t}$ so that the half life of the HFT's position converges to zero as the trading frequency increases. This allows the HFT to achieve the same performance as without inventory aversion in the high-frequency limit.

\begin{figure}
\begin{center}
\includegraphics[width=0.45\textwidth]{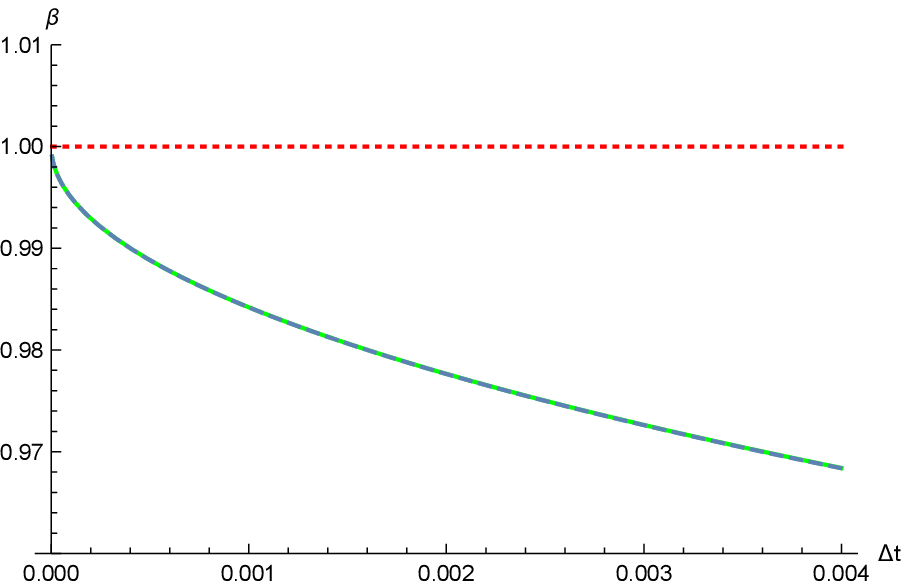}
\includegraphics[width=0.45\textwidth]{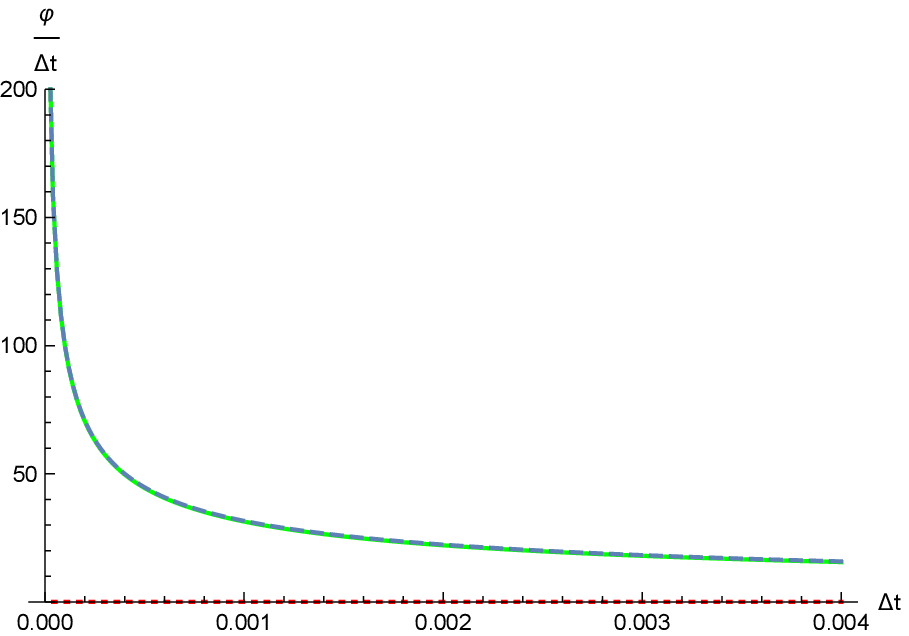}
\caption{\label{fig:monopoly:strategy}
Trading weight $\beta$ for new signals (left panel) and mean reversion speed $\varphi/\Delta t$ (right panel) plotted against the discretization parameter $\Delta t$: high-frequency limit (dotted), first-order approximation (dashed), and exact solution (solid). Note that $0.004 = 1/250$ is approximately one trading day.}
\end{center}
\end{figure}

Figure~\ref{fig:monopoly:strategy} displays the equilibrium signal sensitivity $\beta$ and the mean-reversion rate $\varphi/\Delta t$ as a function of the discretization parameter $\Delta t$. We see that $\beta$ is (in relative terms) close to its high-frequency limit even for rather large values of the discretization parameter $\Delta t$. In contrast, the optimal mean-reversion rate is quite sensitive to $\Delta t$ and explodes quickly as the trading frequency grows. Nevertheless, the first-order expansion for $\varphi$ in Theorem~\ref{thm:monopoly} provides an excellent approximation for the optimal mean-reversion rate.

\paragraph{HFT's optimal performance.}
Unsurprisingly, inventory aversion has a negative effect on the quantity $D$, which measures the optimal performance the HFT can achieve starting from a zero position and no initial signal. The optimal performance is plotted against the discretization parameter $\Delta t$ in Figure~\ref{fig:monopoly:performance}. We observe that the inventory effect is clearly visible for intermediate trading frequencies, but quickly dwindles below 1\% for trading frequencies higher than 100 times per day. At the much shorter time scales corresponding to high frequency trading, the (risk-neutral) high-frequency limit studied directly in \cite{rosu.16} is clearly an excellent approximation.

\begin{figure}
\begin{center}
\includegraphics[width=0.45\textwidth]{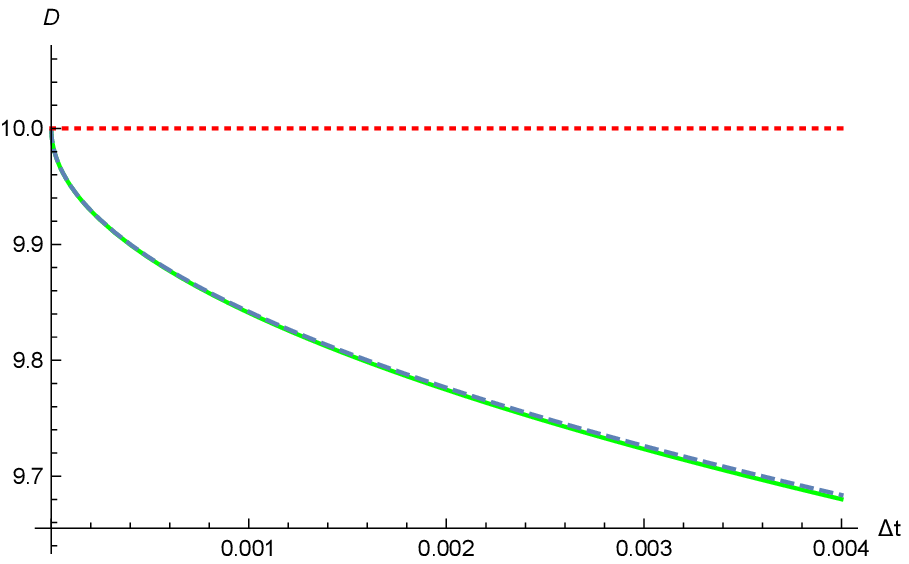}
\includegraphics[width=0.45\textwidth]{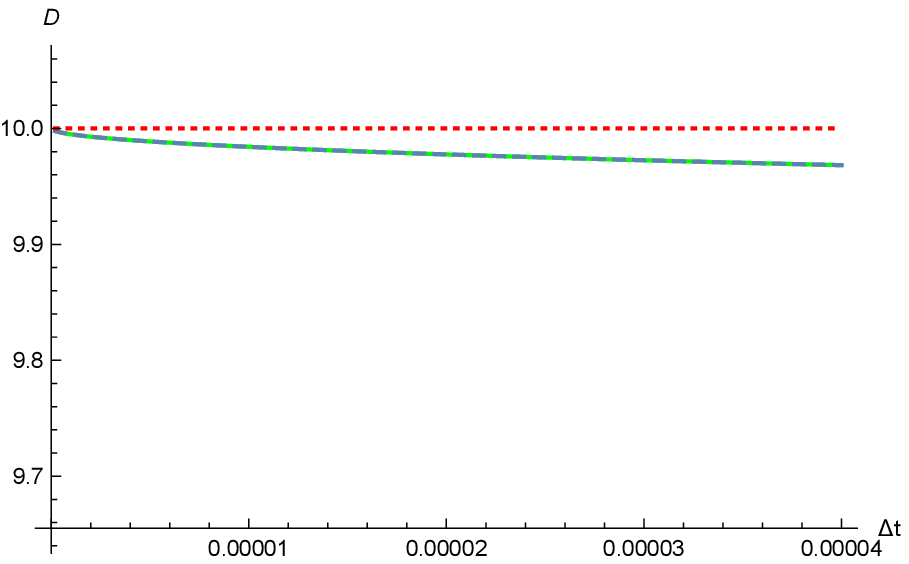}
\caption{\label{fig:monopoly:performance}
HFTs optimal performance $D$ plotted against the discretization parameter $\Delta t$: high-frequency limit (dotted), first-order approximation (dashed), and exact solution (solid). The left panel shows intermediate trading frequencies lower than once per day. The right panel zooms in on higher trading frequencies of more than 100 trades per day. }
\end{center}
\end{figure}

\paragraph{Pricing rule.}
The equilibrium pricing rule adjusts the already revealed part $S_{n-1}$ of the asset's fundamental value linearly for the net order flow $\Delta Y_n$ in trading round $n$ and the HFT's predicted inventory $M_{n-1}$ (cf.~\eqref{eqn:monopoly:price}). At the leading-order, the price impact parameter $\lambda$ is decreasing in the inventory aversion parameter $\gamma$: the HFT's inventory management interferes with the optimal exploitation of the informational advantage that is possible in the risk-neutral case. This in turn allows the dealers to break even on average with a smaller price impact parameter. The equilibrium sensitivity $\mu = \lambda\varphi$ of the execution price with respect to the HFT's predicted inventory is increasing in the inventory aversion parameter $\gamma$ at the order $O(\sqrt{\Delta t})$. The reason is that, for a given predicted inventory (say, positive), a larger risk aversion parameter $\gamma$ implies that the HFT has a stronger incentive to sell. Therefore, the same net order flow is likely to correspond to an even larger buying incentive that needs to be offset by a higher execution price. Finally, observe that $\mu = \lambda\varphi$ leads to an equilibrium execution price 
\begin{align*}
P_n
&= S_{n-1} + \lambda \Delta Y_n + \mu M_{n-1}
= S_{n-1} + \lambda (\Delta K_n + \beta \Delta S_n)
\end{align*}
that does not (functionally) depend on the HFT's predicted inventory.

\begin{figure}
\begin{center}
\includegraphics[width=0.45\textwidth]{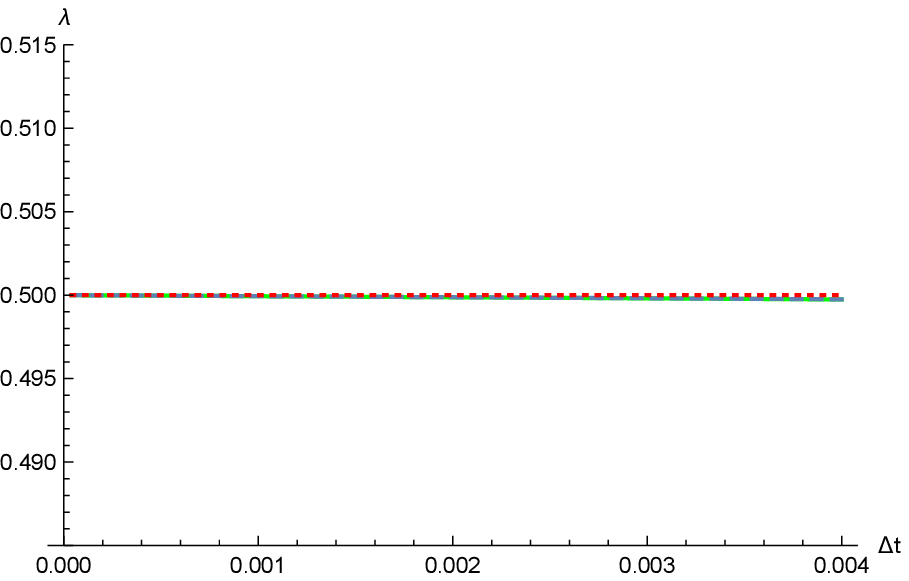}
\includegraphics[width=0.45\textwidth]{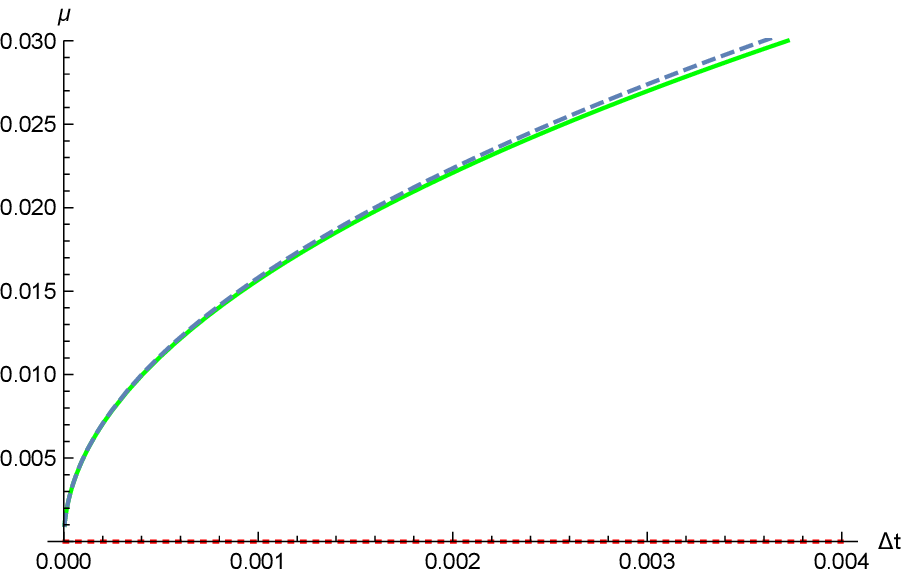}
\caption{\label{fig:monopoly:pricing rule}
Equilibrium pricing rule parameters $\lambda$ (left panel) and $\mu$ (right panel) plotted against the discretization parameter $\Delta t$: high-frequency limit (dotted), first-order approximation (dashed), and exact solution (solid).}
\end{center}
\end{figure}

Figure~\ref{fig:monopoly:pricing rule} shows the parameters $\lambda$ and $\mu$ describing the equilibrium pricing rule as a function of $\Delta t$. For the sensitivity $\lambda$ of the pricing rule with respect to the net order flow, we find that the high-frequency limit, its asymptotic expansion up to terms of order $\Delta t$, and the exact solution of the quartic equation from Theorem~\ref{thm:monopoly} virtually coincide even for trading frequencies as low as once per trading day. This demonstrates that the formula for Kyle's lambda \cite{kyle.85} is very robust with respect to the introduction of inventory aversion. In particular, the apparent welfare differences for noise traders in the present model and in the partial equilibrium setting of \cite{muhlekarbe.webster.17} are negligible in practice. The sensitivity $\mu$ of the execution price with respect to the HFT's predicted inventory is nonzero, unlike its high-frequency limit, but the first-order expansion from Theorem~\ref{thm:monopoly} again provides an excellent approximation even for low trading frequencies.

\begin{remark}
Let us compare our results with the model of \cite{muhlekarbe.webster.17}, who focus on a pricing rule with $\lambda>0$ but $\mu=0$. Accordingly, in their model, the dealers cannot break even in each trading round, but only over the entire time horizon. Despite these differences, the high-frequency limits of $\beta,\lambda$, and $\varphi$ are the same in both models. This provides some justification for focusing on simple pricing rules that do not depend on the HFTs' positions as in \cite{rosu.16,muhlekarbe.webster.17}.

At the next-to-leading order, the effect of the dealers' position-dependent pricing rule becomes visible. Translated into our notation, \cite{muhlekarbe.webster.17} obtains the following asymptotics:
\begin{align*}
\beta
&= \frac{\sigma_K}{\sigma_S} - \frac{3}{4} \left(\frac{\gamma\sigma_K^3}{\sigma_S^3}\right)^{1/2} \sqrt{\Delta t} + O(\Delta t),\\
\lambda
&=  \frac{\sigma_S}{2\sigma_K} - \frac{1}{8}\left(\frac{\gamma\sigma_S}{\sigma_K}\right)^{1/2}\sqrt{\Delta t} + O(\Delta t),\\
\varphi
&= \left(\frac{\gamma\sigma_K}{\sigma_S}\right)^{1/2} \sqrt{\Delta t} + O(\Delta t).
\end{align*}

Both in the partial equilibrium setting of \cite{muhlekarbe.webster.17} and in the present full equilibrium, dealers can increase market depth (i.e., decrease $\lambda$) with an inventory-averse HFT since these have to balance the exploitation of their informational advantage against inventory management. However, the increase in liquidity in our model is asymptotically smaller than in \cite{muhlekarbe.webster.17} ($O(\Delta t)$ in \eqref{eqn:thm:monopoly:lambda} compared to $O(\sqrt{\Delta t})$ in \cite{muhlekarbe.webster.17}). Accordingly, the expected losses of noise traders in our model are comparatively larger. In this sense, noise traders pay the price for enforcing the dealers' zero profit condition in each trading round rather than just over the entire time horizon. 

For the signal sensitivity $\beta$ and the mean-reversion speed $\varphi$, the first-order corrections are of the order $O(\sqrt{\Delta t})$ in both models, but the coefficients of the corresponding correction terms differ. Since the sensitivity to order flow is lower, the HFT exploits her signal about future price changes (slightly) more aggressively than in~\cite{muhlekarbe.webster.17}, in that the $\beta$ coefficient is reduced less compared to the high-frequency limit. As the position-dependent execution prices are higher when the HFT's position is larger, the HFT can also employ a mean-reversion speed $\varphi$ that is asymptotically larger by a factor of $\sqrt{2}$ in our model compared to \cite{muhlekarbe.webster.17}.
\end{remark}

\section{Equilibrium with Competing Insiders}
\label{sec:competition}

\subsection{Nash Competition}

We now turn to the case of $k>1$ possibly heterogeneous HFTs, who compete to exploit their common information about the next-period value increment. Each HFT $i$ has some initial inventory $L_0^i \in \RR$, an individual holding cost $\gamma_i\Delta t/2 >0$ levied on its squared inventory in the risky asset, and an individual discount rate $\rho_i\Delta t \in (0,1)$.

The dealers now clear the market consisting of the $k$ HFTs and the noise traders. They use a separate inventory prediction process for each HFT: for each $i$, the process $M^i = (M^i_n)_{n\geq 0}$ is defined by $M^i_0 = L^i_0$ and
\begin{align}
\label{eqn:competition:inventory prediction process}
\Delta M^i_n
&= \beta_i \Delta S_n - \varphi_iM^i_{n-1},
\end{align}
for some $\vec\beta = (\beta_1,\ldots,\beta_k) \in \RR^k$ and $\vec\varphi=(\varphi_1,\ldots,\varphi_k)\in (0,1]^k$. Consequently, the dealers' linear pricing rule now is of the form
\begin{align*}
P_n
&= S_{n-1} + \lambda\Delta Y_n + \sum_{j=1}^k \mu_j M^j_{n-1},
\end{align*}
for some $\lambda \in \RR$ and $\vec\mu=(\mu_1,\ldots,\mu_k)\in\RR^k$.

As before, risk-neutral HFTs maximize the expectation of their one-period wealth changes. With several HFTs, the cost of purchasing $\Delta L^i_n$ new shares at time $t_n$ for HFT $i$ is
\begin{align*}
P_n \Delta L^i_n
&= \Big(S_{n-1}+\lambda \Delta Y_n + \sum_{j=1}^{k}\mu_j M^j_{n-1}\Big)\Delta L^i_n.
\end{align*}
Here, $\Delta Y_n=\Delta K_n+\sum_{j=1}^k \Delta L^j_n$ is the net order flow in trading round $n$. Comparing this execution price to the HFTs' valuation $S_n$ and taking into account that the noise trades $\Delta K_n$ are independent with mean zero, it follows that, given the other HFTs trade $\Delta L^j_n$, $j \neq i$, the expected change of HFT $i$'s wealth due to her new trade $\Delta L^i_n$ is 
\begin{align*}
\Big(\Delta S_{n}-\lambda\sum_{j=1}^k \Delta L^j_n-\sum_{j=1}^{k}\mu_j M^j_{n-1}\Big)\Delta L^i_n.
\end{align*}

Now, complement this with the inventory penalty, discount, and send the terminal time to infinity. This in turn leads to the following stationary goal functional for HFT $i$:
\begin{align}
\label{eqn:competition:objective}
\EE_1\bigg[ \sum_{n=1}^\infty (1-\rho_i\Delta t)^n \Big\lbrace\Big(\Delta S_n-\lambda\sum_{j=1}^k \Delta L^j_n-\sum_{j=1}^k \mu_j M^j_{n-1}\Big)\Delta L^i_n-\frac{\gamma_i\Delta t}{2}(L^i_{n-1}+\Delta L^i_n)^2\Big\rbrace\bigg],
\end{align}
where $\Delta L^j$, $j\neq i$, are fixed admissible strategies and the optimization runs over $\Delta L^i \in \mathcal{A}$. Our goal is to find an equilibrium in the following sense:

\begin{definition}
\label{def:competition:equilibrium}
Let $(\lambda,\vec\mu,\vec\beta,\vec\varphi)$ be a pricing rule with corresponding inventory prediction processes $M^1,\ldots,M^k$ as in~\eqref{eqn:competition:inventory prediction process}. We say that $(\lambda,\vec\mu,\vec\beta,\vec\varphi)$ forms a \emph{(linear) equilibrium} if:
\begin{enumerate}
\item Given the pricing rule $(\lambda,\vec\mu,\vec\beta,\vec\varphi)$, the strategies $(\Delta M^1,\ldots,\Delta M^k)$ are admissible and form a Nash equilibrium. That is, each HFT $i$ cannot improve her performance by deviating from the strategy $\Delta M^i$ while the other HFTs' strategies $\Delta M^j$, $j\neq i$, remain unchanged.

\item Given that the HFTs use the strategies $(\Delta M^1,\ldots,\Delta M^k)$, the dealers' conditional expected profits in each trading round are zero:
\begin{align}
\label{eqn:def:competition:equilibrium:zero profit condition} 
\wt \EE_n[\Delta S_n]
&= \lambda \Delta Y_n + \sum_{j=1}^k \mu_j M^j_{n-1},\quad n\geq 1,
\end{align}
where $\Delta Y_n = \Delta K_n + \sum_{j=1}^k \Delta M^j_{n}$ is the net order flow in trading round $n$.
\end{enumerate}
\end{definition}

\subsection{Existence and Asymptotics}

The following result identifies a Nash equilibrium for $k$ competing HFTs through a system of nonlinear equations:

\begin{theorem}[Equilibrium]~
\label{thm:competition}
\begin{enumerate}
\item For sufficiently small $\Delta t \geq 0$, the constrained system
\begin{align}
\label{eqn:thm:competition:system:1}
\beta_\Sigma
&= \sum_{j=1}^k \beta_j,\\
\label{eqn:thm:competition:system:2}
0
&= (1-\rho_i\Delta t) \beta_i^2 \beta_\Sigma^2 -(2-\rho_i \Delta  t + \beta_\Sigma  \gamma_i \Delta t)\Big(\frac{\sigma_K}{\sigma_S}\Big)^2  \beta_i \beta_\Sigma + \Big(\frac{\sigma_K}{\sigma_S}\Big)^4 (1-\beta_i \gamma_i \Delta t), \quad i=1,\ldots,k,\\
\label{eqn:thm:competition:system:3}
0
&< \beta_\Sigma,\quad\text{and}\quad
0
< \beta_i \beta_\Sigma \leq \frac{\sigma_K^2}{\sigma_S^2},\quad i=1,\ldots,k,
\end{align}
has a unique solution $(\beta_\Sigma,\vec\beta) = (\beta_\Sigma,\beta_1,\ldots,\beta_k)$. This solution has the following asymptotics as $\Delta t \to 0$:
\begin{align}
\label{eqn:thm:competition:beta}
\beta_i
&= k^{-\frac{1}{2}} \frac{\sigma_K}{\sigma_S} - \frac{(1+k)^{1/2}}{2k^{3/4}} \Big(2\gamma_i^{1/2}-\frac{1}{k}\sum_{j=1}^k \gamma_j^{1/2} \Big) \left(\frac{\sigma_K}{\sigma_S}\right)^{3/2}\sqrt{\Delta t} + O(\Delta t),\\
\label{eqn:thm:competition:beta sum}
\beta_\Sigma
&= k^{\frac{1}{2}} \frac{\sigma_K}{\sigma_S} - \frac{(1+k)^{1/2}}{2k^{3/4}} \sum_{j=1}^k \gamma_j^{1/2} \left(\frac{\sigma_K}{\sigma_S}\right)^{3/2} \sqrt{\Delta t} + O(\Delta t).
\end{align}

\item Define, for sufficiently small $\Delta t \geq 0$,\footnote{Note that $\lambda\beta_\Sigma < 1$, so that $\varphi_i$ is well defined. Moreover, the upper bound on $\beta_i$ in \eqref{eqn:thm:competition:system:3} is equivalent to $\varphi_i$ being nonnegative.}
\begin{align}
\label{eqn:thm:competition:lambda}
\lambda
&= \frac{\beta_\Sigma \sigma_S^2}{\sigma_K^2+\beta_\Sigma^2\sigma_S^2}
= \frac{k^{1/2}}{1+k}\frac{\sigma_S}{\sigma_K}+\frac{k^{1/4}(k-1)}{2(1+k)^{3/2}}\frac{1}{k}\sum_{j=1}^k\gamma_j^{1/2}\left(\frac{\sigma_S}{\sigma_K}\right)^{1/2}\sqrt{\Delta t}+O(\Delta t) ,\\
\label{eqn:thm:competition:phi}
\varphi_i
&= 1- \frac{\lambda\beta_i}{1-\lambda\beta_\Sigma}
= \frac{(1+k)^{1/2}}{k^{1/4}}\gamma_i^{1/2}\left(\frac{\sigma_K}{\sigma_S}\right)^{1/2}\sqrt{\Delta t}+O(\Delta t),\\
\label{eqn:thm:competition:mu}
\mu_i
&= \lambda \varphi_i
= \frac{k^{1/4}}{(1+k)^{1/2}}\gamma_i^{1/2}\left(\frac{\sigma_S}{\sigma_K}\right)^{1/2}\sqrt{\Delta t}+O(\Delta t),
\end{align}
and set $\vec\varphi = (\varphi_1,\ldots,\varphi_k)$ and $\vec\mu = (\mu_1,\ldots,\mu_k)$. Then, for sufficiently small $\Delta t>0$, the pricing rule $(\lambda,\vec\mu,\vec\beta,\vec\varphi)$ forms a linear equilibrium.
\end{enumerate}
\end{theorem}

The optimal performance of HFT $i$ in the equilibrium can be computed explicitly in terms of $\beta_i,\beta_\Sigma$, and the model parameters:

\begin{proposition}[HFT~$i$'s optimization]
\label{prop:HFT}
Let $(\lambda,\vec\mu,\vec\beta,\vec\varphi)$ be as in Theorem~\ref{thm:competition}. Recall that for each $j=1,\ldots,k$, the inventory prediction process $M^j = (M^j_n)_{n\geq 0}$ is defined by $M^j_0 = L^j_0$ and
\begin{align*}
\Delta M^j_n
&= \beta_j\Delta S_n - \varphi_j M^j_{n-1}.
\end{align*}
Fix $i$ and suppose that the dealers use the pricing rule $(\lambda,\vec\mu,\vec\beta,\vec\varphi)$ and that each other HFT $j\neq i$ uses the strategy $\Delta M^j$. Then, for sufficiently small $\Delta t > 0$, the strategy $\Delta M^i$ is optimal for HFT~$i$ and her optimal performance at time $0$ is
\begin{align*}
-\frac{A_i}{2}(L^i_0)^2 + \frac{B_i}{2}(\Delta S_1)^2 - C_iL^i_0\Delta S_1 + D_i,
\end{align*}
where
\begin{align}
\notag
A_i
&= \frac{(1-\rho_i\Delta t)(1-\varphi_i)^2}{1-(1-\rho_i\Delta t)(1-\varphi_i)^2}\gamma_i \Delta t
= \frac{k^{1/4}}{2(1+k)^{1/2}} \gamma_i^{1/2}\left(\frac{\sigma_S}{\sigma_K}\right)^{1/2}\sqrt{\Delta t} + O(\Delta t),\\
\notag
B_i
&= (1-\rho_i\Delta t)\beta_i  (2(1-\lambda\beta_\Sigma)-\beta_i(A_i+ \gamma_i  \Delta t))\\
&= \frac{2}{k^{1/2}(1+k)} \frac{\sigma_K}{\sigma_S} + \frac{1}{2k^{3/4}(1+k)^{3/2}} \Big[(2+6k)\frac{1}{k}\sum_{j=1}^k\gamma_j^{1/2} -5(1+k)\gamma_i^{1/2} \Big] \left(\frac{\sigma_K}{\sigma_S}\right)^{3/2}\sqrt{\Delta t} + O(\Delta t),\notag\\
\notag
C_i
&=(1-\rho_i\Delta t)( \beta_i  (1-\varphi_i) (A_i+\gamma_i  \Delta t)+\varphi_i  (1- \lambda\beta_\Sigma))\\
&= \frac{3}{2k^{1/4}(1+k)^{1/2}} \gamma_i^{1/2} \left(\frac{\sigma_K}{\sigma_S}\right)^{1/2} \sqrt{\Delta t} + O(\Delta t),\notag\\
\label{eqn:prop:HFT:D}
D_i
&= \frac{(1-\rho_i\Delta t)B_i \sigma_S^2}{2\rho_i}\\
&= \frac{1}{k^{1/2}(1+k)} \frac{\sigma_S\sigma_K}{\rho_i} +\frac{1}{4k^{3/4}(1+k)^{3/2}} \Big[(2+6k)\frac{1}{k}\sum_{j=1}^k\gamma_j^{1/2} - 5(1+k)\gamma_i^{1/2} \Big] \frac{\sigma_S^{1/2}\sigma_K^{3/2}}{\rho_i}\sqrt{\Delta t} + O(\Delta t).\notag
\end{align}
\end{proposition}

The proofs of Theorem~\ref{thm:competition} and Proposition~\ref{prop:HFT} are provided in Appendix~\ref{sec:proofs}.

\begin{remark}
\label{rem:stability}
The equilibrium from Theorem~\ref{thm:competition} has a certain stability property with respect to mispredicted inventories. Indeed, the proof of Proposition~\ref{prop:HFT} in Appendix~\ref{sec:proofs} (in particular, Lemma~\ref{lem:DPE}) shows that HFT $i$'s optimal strategy is given by
\begin{align*}
\Delta L^i_n
&= \Delta M^i_n - \zeta_i (L^i_{n-1}-M^i_{n-1}),\quad\text{for some }\zeta_i \in (0,1).
\end{align*}
If the dealers' initial inventory prediction is correct ($M^i_0 = L^i_0$), this shows the optimality of $\Delta L^i = \Delta M^i$ for HFT $i$. But even if the dealers' inventory prediction is incorrect, the HFT has an incentive to gradually reduce the distance between her actual inventory and the prediction:
\begin{align*}
\Delta (L^i - M^i)_n
&= - \zeta_i (L^i_{n-1}-M^i_{n-1}).
\end{align*}
\end{remark}

\begin{remark}
\label{rem:tractability}
In the equilibrium, the goal functional \eqref{eqn:competition:objective} of HFT $i$ simplifies considerably. Indeed, by plugging $\mu_j = \lambda\varphi_j$ and $\Delta L^j_n = \Delta M^j_n = \beta_j\Delta S_j - \varphi_j M^j_{n-1}$ into \eqref{eqn:competition:objective} for $j\neq i$, the inventory prediction processes $M^j$, $j\neq i$, drop out of the optimization criterion. Whence, the value function of HFT $i$ only needs to keep track of the value increments and HFT $i$'s own actual inventory $L^i$ and inventory prediction $M^i$ but not the other HFTs' inventory predictions. This reduces the dimension of the problem from $2+k$ to $3$.
\end{remark}

\subsection{Comparative Statics}

We now discuss the comparative statics of the equilibrium with several competing HFTs. 

\paragraph{High-frequency limit.}
As for the monopolistic case covered by Theorem~\ref{thm:monopoly}, the high-frequency limits of all relevant quantities are very good approximations at the fast trading speeds relevant for high-frequency trading. For the dealers' equilibrium pricing rule, \eqref{eqn:thm:competition:lambda} and \eqref{eqn:thm:competition:mu} show convergence towards the risk-neutral model with imperfect competition studied by \cite{admati.pfleiderer.88,subrahmanyam.98}. Indeed, like for one monopolistic HFT, the sensitivities $\mu_1,\ldots,\mu_k$ of the execution price with respect to the HFTs' positions vanish at rate $\sqrt{\Delta t}$, whereas the sensitivity $\lambda$ with respect to the net order flow converges to its risk-neutral counterpart. Likewise, the weight $\beta_i$ that is placed on new trading signals also converges to its risk-neutral counterpart. As in the monopolistic case, the mean-reversion rate $\varphi_i/\Delta t$ diverges as the trading frequency increases. Consequently, the HFT $i$'s performance $D_i$ converges to its risk-neutral counterpart, as the inventory penalty disappears in the high-frequency limit.

\paragraph{Nash competition.}
To understand how the Nash competition between the HFTs impacts the scaling of the high-frequency limits, it is helpful to compare them to the corresponding results that obtain if the HFTs coordinate their actions through a central planner who aims to maximize their aggregate performance. We focus on the high-frequency limit of $k$ homogeneous HFTs that have the same cost and preference parameters $\gamma$ and $\rho$ and initial inventories. Then, the central planner problem reduces to the case of a monopolistic HFT studied in Section~\ref{sec:monopoly}, but with $\gamma$ replaced by $\gamma/k$. That is, the ``aggregate HFT'' accumulates all of the individual inventory tolerances. Since the limiting price impact and trading strategy in Theorem~\ref{thm:monopoly} are independent of inventory aversion, we observe the typical effects of Nash competition on trading and welfare. Indeed, with imperfect competition, the HFTs overuse their common informational advantage in that their aggregate signal sensitivity is too large: $\beta_\Sigma=k^{1/2}\sigma_K/\sigma_S$ under competition compared to $\beta_\Sigma=\sigma_K/\sigma_S$ with coordination. As observed by \cite{subrahmanyam.98}, this is an incarnation of the classical ``tragedy of the commons'': the HFTs overuse their common good, the liquidity available in the market. Put differently, each of them only internalizes their own price impact cost, but not the negative effect this has for the others. This excess trading volume (with the same informational advantage) drives down the HFTs' aggregate performance by a factor of $2k^{1/2}/(1+k) \in (0,1]$. This in turn allows the dealers to break even with a smaller price impact parameter $\lambda$; it is reduced by the same factor $2k^{1/2}/(1+k) \in (0,1]$.

\paragraph{Inventory aversion.}
Under Nash competition, the first-order correction term of the price impact parameter $\lambda$ (cf.~\eqref{eqn:thm:competition:lambda}) is of the order $O(\sqrt{\Delta t})$ and \emph{increases} with the inventory aversion parameter $\gamma$. This is in stark contrast to the monopolistic case where the first-order correction term is of the order $O(\Delta t)$ and decreasing in $\gamma$ (cf.~\eqref{eqn:thm:monopoly:lambda}). This initially surprising effect is explained by the excessive trading due to Nash competition. Indeed, with inventory aversion, the HFTs scale back their aggregate trades on their signals as described by $\beta_\Sigma$ in \eqref{eqn:thm:competition:beta sum}. This moves the HFTs closer to their coordinated equilibrium and forces the dealers to recuperate their lost trading profits by increasing $\lambda$.

Regarding an HFT's individual performance, increasing inventory aversion has two opposing effects. On the one hand, increasing inventory aversion of course lowers the HFT's performance through the higher inventory penalty. On the other hand, as above, it moves the competitive HFTs closer to their coordinated equilibrium and thereby improves each HFT's individual performance. This counterplay is reflected in the first-order correction term of an HFT's performance $D_i$ in \eqref{eqn:prop:HFT:D}. For example, in the case of homogeneous HFTs, the first-order correction term is negative for $k\leq 2$, positive for $k\geq 4$, and vanishes for $k=3$. Whence, under sufficiently strong competition, inventory aversion has a beneficial effect on the HFTs' performance compared to the corresponding risk-neutral equilibrium.

\paragraph{Heterogeneity.}
While the aggregate signal sensitivity $\beta_\Sigma$ is always decreasing in the individual HFTs' inventory aversions, the sign of the first-order correction of the individual signal sensitivity $\beta_i$ (cf.~\eqref{eqn:thm:competition:beta}) of HFT $i$ depends on how her inventory aversion relates to those of the other HFTs. Indeed, if an HFT's inventory costs are sufficiently low compared to the others', then imperfect competition allows to exploit price signals even more aggressively than in the risk-neutral case, thereby also improving the respective performance (cf.~\eqref{eqn:prop:HFT:D}). For HFTs with comparatively high inventory costs, the situation is reversed.

\begin{figure}
\begin{center}
\includegraphics[width=0.45\textwidth]{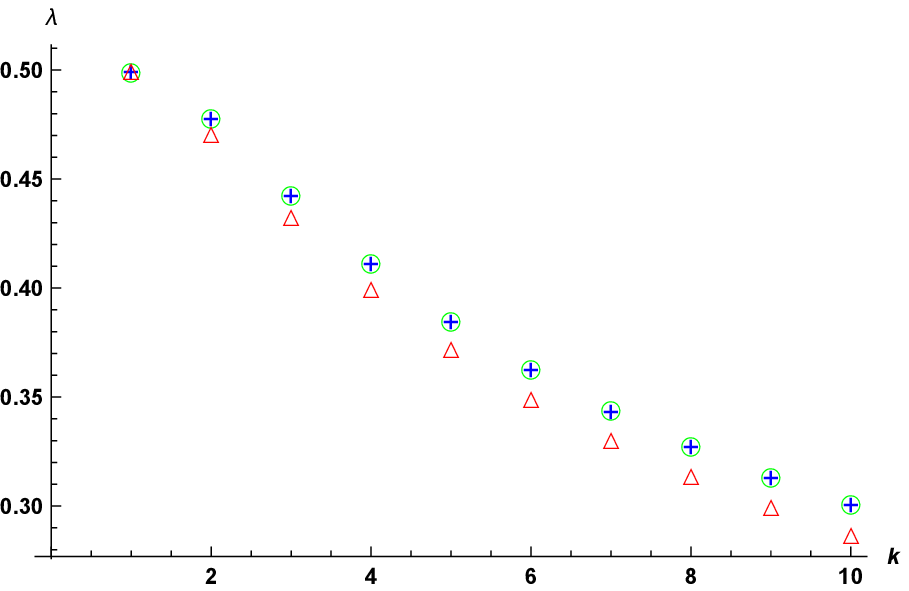}
\includegraphics[width=0.45\textwidth]{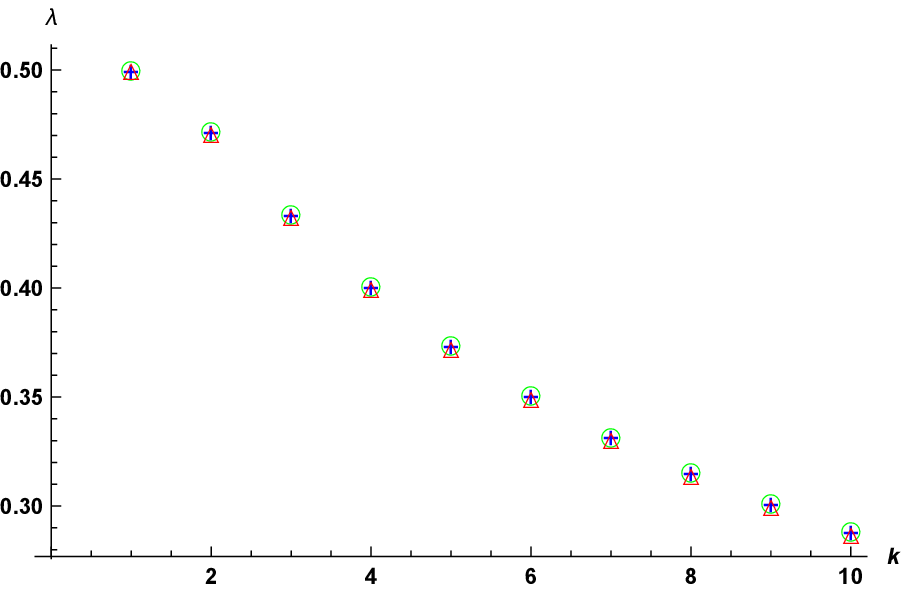}
\caption{\label{fig:competition}
Sensitivity $\lambda$ of the execution price with respect to the net order flow plotted against the number $k$ of competing HFTs, for 1 trading round per day (left panel) and 100 trading rounds per day (right panel): high-frequency limit (triangles), first-order approximation (crosses), and exact solution (circles).}
\end{center}
\end{figure}

\medskip

The effects described above are interesting from a theoretical point of view. However, Figure~\ref{fig:competition} illustrates that while the impact of inventory aversion is clearly visible at low trading frequencies (e.g., daily), it disappears quickly as $\Delta t$ tends to zero. Accordingly, these effects only play a secondary role in a high-frequency context.

\subsection{Transaction Taxes}

Transaction taxes are often mentioned as a possible tool to improve market quality by curbing high-frequency trading.  As in the risk-neutral one-period model of \cite{subrahmanyam.98}, quadratic transaction taxes can also be incorporated into the present framework.\footnote{Other specifications such as taxes proportional to the trade size are not tractable because they lead to nonlinear filtering problems. However, we expect the broad conclusions for such models to be similar.} This means that HFTs incur an additional transaction cost that is proportional to the squared sizes of their individual trades. The stationary goal functional~\eqref{eqn:competition:objective} then becomes
\begin{align*}
\EE_1\bigg[ \sum_{n=1}^\infty (1-\rho_i\Delta t)^n \Big\lbrace\Big(\Delta S_n-c\Delta L_n^i-\lambda\sum_{j=1}^k \Delta L^j_n-\sum_{j=1}^k \mu_j M^j_{n-1}\Big)\Delta L^i_n-\frac{\gamma_i\Delta t}{2}(L^i_{n-1}+\Delta L^i_n)^2\Big\rbrace\bigg]
\end{align*}
for some transaction tax parameter $c>0$.

In analogy to Theorem~\ref{thm:competition}, one can show that an equilibrium pricing rule is identified by the solution of $(\beta_\Sigma,\beta_1,\ldots,\beta_k)$ of the constrained system\footnote{While $\lambda$ and $\mu_i$ are still given by \eqref{eqn:thm:competition:lambda} and \eqref{eqn:thm:competition:mu}, respectively, $\varphi_i$ is determined by $\varphi_i = 1 - \frac{(\lambda + 2c)\beta_i}{1-\lambda \beta_\Sigma}$. The upper bound on $\beta_i$ is equivalent to $\varphi_i$ being nonnegative.}
\begin{align*}
\beta_\Sigma
&= \sum_{j=1}^k \beta_j,\\
0
&= (1-\rho_i\Delta t) \beta_i^2 \bigg[\beta_\Sigma+2c\bigg(\Big(\frac{\sigma_K}{\sigma_S}\Big)^2+\beta_\Sigma^2\bigg)\bigg]^2 \notag\\ 
&\quad-\left[\bigg(2c\Big(\Big(\frac{\sigma_K}{\sigma_S}\Big)^2+\beta_\Sigma^2\Big)+\beta_\Sigma\bigg)(2-\rho_i\Delta t)+\beta_\Sigma^2\gamma_i\Delta t\right]\Big(\frac{\sigma_K}{\sigma_S}\Big)^2  \beta_i\\ 
&\quad+ \Big(\frac{\sigma_K}{\sigma_S}\Big)^4 (1-\beta_i \gamma_i \Delta t), \quad i=1,\ldots,k,\notag\\
0
&<\beta_\Sigma,
\quad\text{and}\quad
0
< \beta_i \leq \frac{\sigma_K^2}{2c(\sigma_K^2 +\beta_\Sigma^2\sigma_S^2) + \beta_\Sigma \sigma_S^2},\quad i=1,\ldots,k.
\end{align*}

\begin{figure}
\begin{center}
\includegraphics[width=0.45\textwidth]{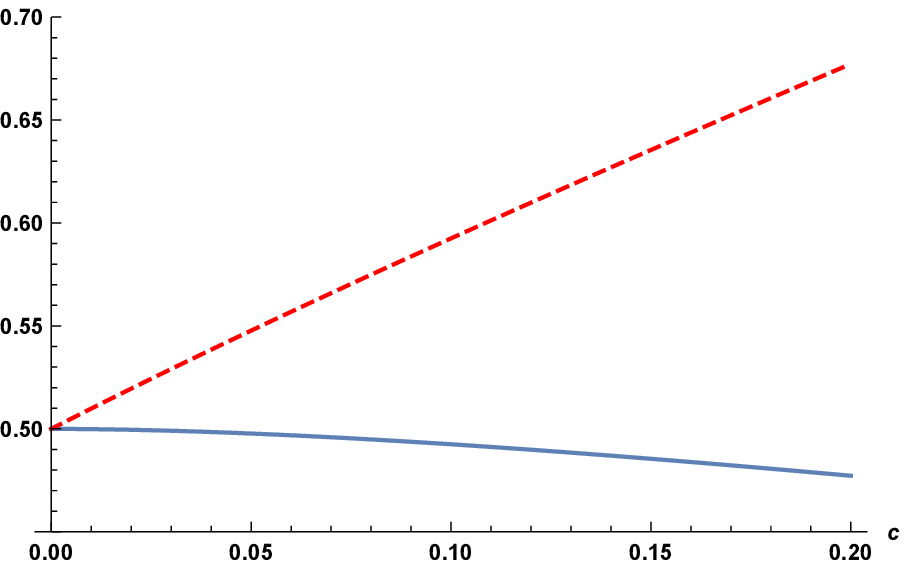}
\includegraphics[width=0.45\textwidth]{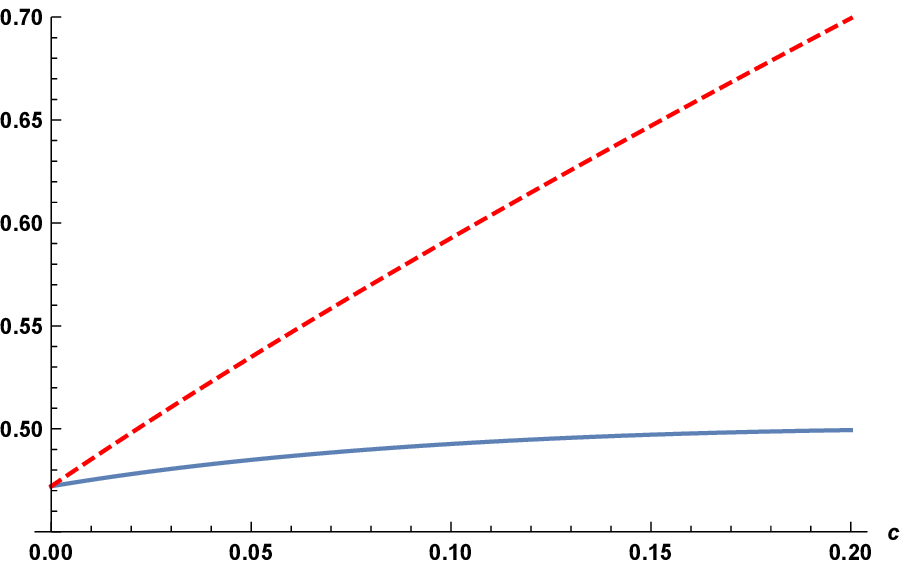}
\caption{\label{fig:tax}
Total trading costs $\lambda + c$ (dashed) and price impact $\lambda$ (solid) with respect to the net order flow plotted against the transaction tax parameter $c$, for a monopolistic HFT (left panel) and  two competitors (right panel). The trading frequency is 100 times per day.
}
\end{center}
\end{figure}

The corresponding price impact parameter $\lambda$ and total trading costs $\lambda + c$ are depicted in Figure~\ref{fig:tax}, both for a monopolistic HFT and for two competitors. Our results corroborate the findings of \cite{subrahmanyam.98} in a one-period model with risk-neutral HFTs. The total trading costs $\lambda + c$ are in both cases increasing in the transaction tax $c$. Regarding price impact only, the comparative statics are different in the monopolistic and the oligopolistic case. For a monopolistic insider, a small transaction tax decreases price impact. In contrast, with Nash competition, transaction taxes tend to \emph{increase} price impact. The reason is again the negative externality inherent in the HFTs choices: ``without a transaction tax, [\ldots] they end up trading `too much' in equilibrium, that is, in a dissipative fashion such that their profits decrease in the total number of informed agents in the market, leading to greater market liquidity. A transaction tax causes them to scale back their trading to the extent that, while their profits net of the transaction tax are decreasing in the tax, the profits gross of the transaction tax are increasing in the tax. This causes the transaction tax to have a perverse effect: it reduces market liquidity (and increases the adverse price impact faced by informationless traders), but it also reduces informed trader profits''~\cite{subrahmanyam.98}.

This basic mechanism should be kept in mind when discussing transaction taxes in the context of high-frequency trading. In the modeling framework considered here, market liquidity cannot be improved by taxation but only by encouraging more competition among HFTs, compare Figure~\ref{fig:competition}. Taxes can, however, become socially preferable if one considers costs for information acquisition~\cite{subrahmanyam.98} or (over-)investment in trading technologies~\cite{biais.al.15}. Incorporating such features into the present model is a challenging but important direction for future research.

\appendix

\section{Proofs}
\label{sec:proofs}

In this section, we prove Theorem~\ref{thm:competition} and Proposition~\ref{prop:HFT}. We start with the existence of a local solution to the system \eqref{eqn:thm:competition:system:1}--\eqref{eqn:thm:competition:system:3} for small $\Delta t$.

\begin{lemma}
\label{lem:system}
There is $\varepsilon > 0$ such that for all $\Delta t \in [0,\varepsilon)$, the system \eqref{eqn:thm:competition:system:1}--\eqref{eqn:thm:competition:system:3} has a unique solution $(\beta_\Sigma,\beta_1,\ldots,\beta_k)$ with the asymptotics \eqref{eqn:thm:competition:beta}--\eqref{eqn:thm:competition:beta sum}.
\end{lemma}

\begin{proof}
\emph{Step 1.} We note that for $\Delta t = 0$, the system \eqref{eqn:thm:competition:system:1}--\eqref{eqn:thm:competition:system:3} simplifies to
\begin{align*}
\beta_\Sigma
= \sum_{j=1}^k\beta_j,\quad
\beta_i \beta_\Sigma
= \left(\frac{\sigma_K}{\sigma_S}\right)^2,\quad
0
< \beta_\Sigma,\quad
0
< \beta_i\beta_\Sigma
\leq \frac{\sigma_K^2}{\sigma_S^2}, \quad i=1,\ldots,k,
\end{align*}
which has the unique solution $\bar\beta_\Sigma = k^{\frac{1}{2}}\frac{\sigma_K}{\sigma_S}$, $\bar\beta_i = k^{-\frac{1}{2}}\frac{\sigma_K}{\sigma_S}$, $i=1,\ldots,k$.

\emph{Step 2.} We next show the existence of a solution for small $\Delta t > 0$. First, we transform the system \eqref{eqn:thm:competition:system:1}--\eqref{eqn:thm:competition:system:3} to an equivalent system which is amenable to the implicit function theorem. Consider the following reparameterization of the domain $(0,\infty)\times\RR^{1+k}$ of the variables $(\Delta t,\beta_\Sigma,\beta_1,\ldots,\beta_k)$:
\begin{align*}
\delta
&= \sqrt{\Delta t},\quad
x = \frac{1}{\delta} \left(\beta_\Sigma-\bar\beta_\Sigma \right),\quad
y_i
=\frac{1}{\delta} \left(\beta_i - \bar\beta_i \right),\quad i=1,\ldots,k.
\end{align*}
After inserting this change of variables into \eqref{eqn:thm:competition:system:1}--\eqref{eqn:thm:competition:system:3}, simplifying, and multiplying the resulting equations by convenient nonzero terms, it follows that for any $0<\Delta t = \delta^2$, the original system \eqref{eqn:thm:competition:system:1}--\eqref{eqn:thm:competition:system:3} is equivalent to the system
\begin{align}
\label{eqn:lem:system:pf:transformed system:equalities}
\begin{split}
0
&= h_0(\delta,x,y_1,\ldots,y_k),\\
0
&= h_i(\delta,x,y_1,\ldots,y_k),\quad i=1,\ldots,k,
\end{split}\\
\label{eqn:lem:system:pf:transformed system:inequality}
\begin{split}
-\bar\beta_\Sigma
&< \delta x,\\
-\frac{\sigma_K^2}{\sigma_S^2}
&< \delta \bar\beta_i(x + k y_i) + \delta^2 xy_i 
\leq 0,\quad i=1,\ldots,k,
\end{split}
\end{align}
where
\begin{align*}
h_0(\delta,x,y_1,\ldots,y_k)
&= x - \sum_{j=1}^k y_j,\\
h_i(\delta,x,y_1,\ldots,y_k)
&= (x+k y_i)^2- k^\frac{1}{2} (1+k) \gamma_i \left(\frac{\sigma_K}{\sigma_S}\right)^3+ (\ast)\,\delta  + (\ast)\,\delta^2 + (\ast)\,\delta^3 + (\ast)\,\delta^4.
\end{align*}
Here, the $(\ast)$-terms stand for generic polynomials in $x,y_1,\ldots,y_k$, which do not depend on $\delta$ and are not important for the subsequent calculations.

Second, we show that the transformed system \eqref{eqn:lem:system:pf:transformed system:equalities}--\eqref{eqn:lem:system:pf:transformed system:inequality} has a solution $(x(\delta),y_1(\delta),\ldots,y_k(\delta))$ in a neighborhood of $\delta = 0$. It is readily verified that
\begin{align}
\label{eqn:lem:system:pf:first-order terms}
\begin{split}
\bar{x}
&=-\frac{(1+k)^{1/2}}{2k^{3/4}}\sum_{j=1}^k\gamma_j^{1/2}\left(\frac{\sigma_K}{\sigma_S}\right)^{3/2},\\
\bar{y}_i
&=-\frac{(1+k)^{1/2}}{2k^{3/4}}\left(2\gamma_i^{1/2}-\frac{1}{k}\sum_{j=1}^k \gamma_j^{1/2}\right)\left(\frac{\sigma_K}{\sigma_S}\right)^{3/2}, \quad i=1,\ldots,k,
\end{split}
\end{align}
is a solution to the quadratic system that arises from \eqref{eqn:lem:system:pf:transformed system:equalities} by inserting $\delta = 0$. Moreover, the Jacobian of $(h_0,h_1,\ldots,h_k)$ with respect to the variables $(x,y_1,\ldots,y_k)$, evaluated at $\delta = 0$ and $(x,y_1,\ldots,y_k) = (\bar x, \bar y_1,\ldots,\bar y_k)$, is 
\begin{align*}
\begin{pmatrix}
1      & -1         & -1         & \cdots & -1\\
a_1    & ka_1       &            &        &   \\
a_2    &            & ka_2       &        &   \\
\vdots &            &            & \ddots &   \\
a_k    &            &            &        & ka_k
\end{pmatrix},
\end{align*}
where $a_i = 2(\bar x+k\bar y_i) < 0$ and zero entries are omitted. Using row (or column) transformations, one can verify that this matrix is invertible. Therefore, the implicit function theorem yields an $\varepsilon' >0$ and a continuously differentiable function $(x, y_1,\ldots, y_k):(-\varepsilon',\varepsilon') \to \RR^{1+k}$ such that $(x(\delta),y_1(\delta),\ldots, y_k(\delta))$ solves \eqref{eqn:lem:system:pf:transformed system:equalities} for all $\delta\in(-\varepsilon',\varepsilon')$ and $x(0) = \bar x$ and $y_i(0) = \bar y_i$ for $i=1,\ldots,k$. Since $\bar x + k \bar y_i < 0$, making $\varepsilon'$ smaller if necessary, we can also ensure that $(x(\delta),y_1(\delta),\ldots, y_k(\delta))$ satisfies the inequalities \eqref{eqn:lem:system:pf:transformed system:inequality} for all $\delta \in (-\varepsilon',\varepsilon')$.

Third, after reverting the change of variables and setting $\varepsilon := (\varepsilon')^2$, we can conclude that for any $\Delta t \in [0,\varepsilon)$,
\begin{align*}
\beta_\Sigma(\Delta t)
&:= k^{\frac{1}{2}}\frac{\sigma_K}{\sigma_S} + x(\sqrt{\Delta t}) \sqrt{\Delta t},\quad
\beta_i(\Delta t)
:= k^{-\frac{1}{2}}\frac{\sigma_K}{\sigma_S} + y_i(\sqrt{\Delta t}) \sqrt{\Delta t},\quad i=1,\ldots,k,
\end{align*}
defines a solution to the original system \eqref{eqn:thm:competition:system:1}--\eqref{eqn:thm:competition:system:3}. Moreover, the functions $\beta_\Sigma$, $\beta_i$, $i=1,\ldots,k$, are continuously differentiable on $(0,\varepsilon)$ and continuous on $[0,\varepsilon)$ and, in view of the expressions \eqref{eqn:lem:system:pf:first-order terms}, have the asymptotic expansions \eqref{eqn:thm:competition:beta}--\eqref{eqn:thm:competition:beta sum}.\footnote{Asymptotic expansions for $\beta_\Sigma,\beta_1,\ldots,\beta_k$ up to order $O(\Delta t)$ can be obtained by computing the derivatives $x'(0), y_1'(0),\ldots,y_k'(0)$ by means of the implicit function theorem.}

\emph{Step 3.} We finally address the uniqueness of a solution for small $\Delta t>0$. Denote the right-hand side of \eqref{eqn:thm:competition:system:2} by $g(\beta_\Sigma,\beta_i,\Delta t)$, i.e.,
\begin{align*}
g(\beta_\Sigma,\beta_i,\Delta t)
&= (1-\rho_i\Delta t) \beta_i^2 \beta_\Sigma^2 -(2-\rho_i \Delta  t + \beta_\Sigma  \gamma_i \Delta t)\Big(\frac{\sigma_K}{\sigma_S}\Big)^2  \beta_i \beta_\Sigma + \Big(\frac{\sigma_K}{\sigma_S}\Big)^4 (1-\beta_i \gamma_i \Delta t).
\end{align*}

First, we show that for any $\Delta t>0$ and $\beta_\Sigma >0$, there is a unique $\wh\beta_i = u_i(\beta_\Sigma,\Delta t)$ satisfying $g(\beta_\Sigma,\wh\beta_i,\Delta t) = 0$ and $0<\wh\beta_i\beta_\Sigma < \sigma_K^2/\sigma_S^2$. For any $\Delta t > 0$ and $\beta_\Sigma>0$, the discriminant of the quadratic function $g(\beta_\Sigma,\cdot,\Delta t)$ is positive:
\begin{align*}
4\beta_\Sigma\gamma_i\left(\frac{\sigma_K^2}{\sigma_S^2}+\beta_\Sigma^2\right)\frac{\sigma_K^4}{\sigma_S^4}\Delta t + \Big(\gamma_i \frac{\sigma_K^2}{\sigma_S^2} + \beta_\Sigma(\beta_\Sigma \gamma_i - \rho_i)\Big)^2 \frac{\sigma_K^4}{\sigma_S^4} (\Delta t)^2 > 0.
\end{align*}
Hence, $g(\beta_\Sigma,\cdot,\Delta t)$ has two distinct real roots. Since furthermore
\begin{align*}
g(\beta_\Sigma,0,\Delta t)
&=\frac{\sigma_K^4}{\sigma_S^4}>0,\\
g\Big(\beta_\Sigma,\frac{\sigma_K^2}{\beta_\Sigma \sigma_S^2},\Delta t \Big)
&=-\gamma_i\left(\frac{\sigma_K^2}{\sigma_S^2}+\beta_\Sigma^2\right)\frac{\sigma_K^4}{\beta_\Sigma\sigma_S^4}\Delta t<0,\\
\lim_{\beta_i\to\infty}
g\Big(\beta_\Sigma,\beta_i,\Delta t \Big)
&=\infty,
\end{align*}
the assertion follows from the intermediate value theorem.

Next, we argue that for any $\Delta t > 0$ such that $\rho_i\Delta t < 1$, the function $u_i(\cdot,\Delta t)$ is decreasing on $(0,\infty)$. As $u_i(\beta_\Sigma,\Delta t)$ is the (smaller) solution of a quadratic equation, we have an explicit formula for $u_i(\beta_\Sigma,\Delta t)$. By direct computations and simplifications, one can show that its partial derivative $\frac{\partial u_i}{\partial \beta_\Sigma}(\beta_\Sigma,\Delta t)$ is negative for $\beta_\Sigma > 0$ and $0<\rho_i\Delta t < 1$.

Finally, fix $\Delta t > 0$ small enough and let $(\beta_\Sigma,\beta_1,\ldots,\beta_k)$ and $(\beta_\Sigma',\beta_1',\ldots,\beta_k')$ be two solutions to \eqref{eqn:thm:competition:system:1}--\eqref{eqn:thm:competition:system:3}.
In particular, $\beta_i = u_i(\beta_\Sigma)$ and $\beta_i' = u_i(\beta_\Sigma')$ for all $i=1,\ldots,k$. We may assume without loss of generality that $\beta_\Sigma \leq \beta_\Sigma'$. Then, by the above, $\beta_i' =u_i(\beta_\Sigma') \leq u_i(\beta_\Sigma)=\beta_i$ for $i=1,\ldots,k$. Therefore, $\beta_\Sigma' = \beta_1'+\cdots+\beta_k' \leq \beta_1+\cdots+\beta_k = \beta_\Sigma$. We conclude that $\beta_\Sigma = \beta_\Sigma'$ and $\beta_i = \beta_i'$ for all $i=1,\ldots,k$.
\end{proof}

Let $(\vec\beta,\beta_\Sigma) = (\beta_1,\ldots,\beta_k,\beta_\Sigma)$ be as in Theorem~\ref{thm:competition}~(i). We now turn to the second part of Theorem~\ref{thm:competition}. We first note that $\varphi_i = 1-\frac{\lambda\beta_i}{1-\lambda\beta_\Sigma}$ is well defined since $\lambda\beta_\Sigma \in[0,1)$ by the definition of $\lambda$. Moreover, writing $\varphi_i = \frac{1-\lambda(\beta_\Sigma+\beta_i)}{1-\lambda\beta_\Sigma}$ shows that $\varphi_i < 1$. The expansions \eqref{eqn:thm:competition:lambda}--\eqref{eqn:thm:competition:mu} follow from the corresponding expansions for $\beta_i$ and $\beta_\Sigma$. In particular, the expansion for $\varphi_i$ shows that $\varphi_i > 0$ for $\Delta t>0$ small enough. In view of Lemma~\ref{lem:admissible:general}, this implies the admissibility of the strategies $\Delta M^i$ defined in \eqref{eqn:competition:inventory prediction process}:

\begin{lemma}
\label{lem:admissible:candidate}
For $\Delta t>0$ small enough, $\Delta M^i \in \mathcal{A}$ for all $i=1,\ldots,k$.
\end{lemma}

Next, we show that for the pricing rule and strategies defined in Theorem~\ref{thm:competition}, the dealers' zero profit condition holds.

\begin{lemma}
\label{lem:zero profit}
Suppose that the dealers use the pricing rule $(\lambda,\vec\mu,\vec\beta,\vec\varphi)$ defined in Theorem~\ref{thm:competition} and that each HFT $i$ uses the strategy $\Delta M^i$ as defined in \eqref{eqn:competition:inventory prediction process}. Then the zero profit condition \eqref{eqn:def:competition:equilibrium:zero profit condition} holds.
\end{lemma}

\begin{proof}
In view of the definition of $\Delta M^i$ in \eqref{eqn:competition:inventory prediction process},
\begin{align*}
\Delta Y_n
&= \Delta K_n + \sum_{j=1}^k \Delta M^j_n
= \Delta K_n + \beta_\Sigma \Delta S_n - \sum_{j=1}^k \varphi_j M^j_{n-1},
\end{align*}
where $\beta_\Sigma = \sum_{j=1}^k \beta_j$. Since $\Delta K_n$ and $\Delta S_n$ are independent normally distributed random variables, the random vector $(\Delta S_n, X_n)$, where
\begin{align*}
X_n
&:= \Delta Y_n + \sum_{j=1}^k \varphi_j M^j_{n-1}
= \Delta K_n + \beta_\Sigma \Delta S_n,
\end{align*}
has a bivariate normal distribution. Moreover, $X_n$ is observable with respect to the dealers information set at time $n$. Therefore, by the formula for the conditional marginal mean of a bivariate normal random vector, we find that
\begin{align*}
\wt \EE_n[\Delta S_n]
&= \frac{\mathrm{Cov}(\Delta S_n,X_n)}{\mathrm{Var}(X_n)}X_n
= \frac{\beta_\Sigma \sigma_S^2}{\sigma_K^2 + \beta_\Sigma^2 \sigma_S^2}\big(\Delta Y_n + \sum_{j=1}^k \varphi_j M^j_{n-1}\big).
\end{align*}
The zero profit condition \eqref{eqn:def:competition:equilibrium:zero profit condition} now follows from the definitions of $\lambda$ and $\mu_j$ in \eqref{eqn:thm:competition:lambda} and \eqref{eqn:thm:competition:mu}.
\end{proof}

To complete the proof of Theorem~\ref{thm:competition}~(ii), we need to show that given the dealers' pricing rule $(\lambda,\vec\mu,\vec\beta,\vec\varphi)$, the strategies $(\Delta M^1,\ldots,\Delta M^k)$ form a Nash equilibrium for the HFTs. So fix $i\in\lbrace 1,\ldots,k\rbrace$ and suppose that the dealers use the pricing rule $(\lambda,\vec\mu,\vec\beta,\vec\varphi)$ and that every other HFT $j$, $j \neq i$, uses the strategy 
\begin{align*}
\Delta M^j_n
&= \beta_j \Delta S_n - \varphi_j M^j_{n-1}.
\end{align*}
Plugging these strategies for $j \neq i$ into the goal functional \eqref{eqn:competition:objective} of HFT $i$ and using that $\mu_j = \lambda \varphi_j$ by the definition of $\mu_j$ in \eqref{eqn:thm:competition:mu}, we see that the inventories of the other HFTs disappear:
\begin{align}
\label{eqn:competition:objective:equilibrium:M L}
\EE_1\bigg[ \sum_{n=1}^\infty (1-\rho_i\Delta t)^n \Big\lbrace\big((1-\lambda{\textstyle\sum_{j\neq i}\beta_j})\Delta S_n-\lambda\Delta L^i_n-\mu_i M^i_{n-1}\big)\Delta L^i_n-\frac{\gamma_i\Delta t}{2}(L^i_{n-1}+\Delta L^i_n)^2\Big\rbrace\bigg].
\end{align}
To wit, the individual optimization problem of HFT $i$ in our equilibrium of $k$ competitive HFTs reduces to that of a single HFT who is facing the pricing rule $(\lambda,\mu_i,\beta_i,\varphi_i)$ and is trading a risky asset whose standard deviation of value increments is changed by a factor of $1-\lambda\sum_{j\neq i}\beta_j$.

It will be convenient to represent HFT $i$'s strategy $\Delta L^i$ relative to its prediction $\Delta M^i$. We thus consider a new state variable $Z^i$, defined via
\begin{align}
\label{eqn:competition:Z}
L^i_n
&= M^i_n + Z^i_n,
\end{align}
that keeps track of the deviation of the actual inventory $L^i$ of HFT $i$ from the corresponding prediction $M^i$. It is clearly equivalent to control either $\Delta L^i\in\mathcal{A}$ or $\Delta Z^i\in\mathcal{A}$ (note that $\mathcal{A}$ is a vector space and that $\Delta M^i \in \mathcal{A}$ by Lemma~\ref{lem:admissible:candidate}), and we need to show that $\Delta Z^i\equiv 0$ is optimal. Substituting \eqref{eqn:competition:Z} into \eqref{eqn:competition:objective:equilibrium:M L} and using again the definitions of $\Delta M^i$ and $\mu_i$ yields the following reduced optimization problem:
\begin{align}
\label{eqn:competition:objective:equilibrium:reduced}
\mathcal{J}(M^i_0,\Delta S_1,Z^i_0;\Delta Z^i)
= \EE_1\bigg[ \sum_{n=1}^\infty (1-\rho_i\Delta t)^n f(M^i_{n-1},\Delta S_n, Z^i_{n-1};\Delta Z^i_n)\bigg]\to \max_{\Delta Z^i \in \mathcal{A}}!
\end{align}
where $f:\RR^4\to\RR$ is given by
\begin{align*}
f(M,\Delta S,Z;\Delta Z)
&=
(\eta\Delta S -\lambda \Delta Z)(\beta_i\Delta S - \varphi_i M + \Delta Z) - \frac{\gamma_i\Delta t}{2}(\beta_i\Delta S + (1-\varphi_i)M+Z + \Delta Z)^2
\end{align*}
and
\begin{align}
\label{eqn:eta}
\eta
&= 1-\lambda\beta_\Sigma.
\end{align}

The next lemma provides the value function for \eqref{eqn:competition:objective:equilibrium:reduced} and shows that the optimal feedback control is of the form $\Delta Z^i = -\zeta_i Z^i$ for some $\zeta_i\in(0,1)$. In particular, since $M^i_0 = L^i_0$, we have $Z^i_0 = 0$, so that $\Delta Z^i \equiv 0$ is optimal for \eqref{eqn:competition:objective:equilibrium:reduced}.\footnote{Note that HFT $i$'s optimal strategy is mean-reverting to zero in $Z^i$. To wit, if the dealers' initial inventory prediction is incorrect, the HFT has an incentive to gradually align her actual inventory with the dealers' prediction. In this sense, the equilibrium is stable with respect to mispredicted inventories.} As a consequence, $\Delta L^i = \Delta M^i$ is an optimal strategy for HFT $i$. This proves Proposition~\ref{prop:HFT} and completes, together with Lemmas~\ref{lem:admissible:candidate}--\ref{lem:zero profit}, the proof of Theorem~\ref{thm:competition} (ii).\footnote{The asymptotic expansions for $\lambda$, $\varphi_i$, and $\mu_i$ follow directly from \eqref{eqn:thm:competition:beta}--\eqref{eqn:thm:competition:beta sum}.}

\begin{lemma}
\label{lem:DPE}
Fix $i \in \lbrace 1,\ldots,k \rbrace$, and let $\beta_i,\beta_\Sigma,\lambda,\varphi_i,A_i,B_i,C_i,D_i$ be defined as in Theorem~\ref{thm:competition} and Proposition~\ref{prop:HFT} (for $\Delta t>0$ sufficiently small). 
\begin{enumerate}
\item There are unique $E_i>0$ and $\zeta_i >0$ such that
\begin{align}
\label{eqn:lem:DPE:E and zeta}
\frac{E_i}{1-\rho_i\Delta t}
&= 2\lambda \zeta_i
\quad\text{and}\quad
\zeta_i
= \frac{E_i+\gamma_i\Delta t}{E_i+\gamma_i\Delta t+2\lambda}.
\end{align}

\item For $\Delta t>0$ sufficiently small, define $F_i$ and $G_i$ by
\begin{align}
\label{eqn:lem:DPE:F}
\frac{F_i}{1-\rho_i\Delta t}
&= \frac{\lambda \varphi_i\zeta_i+(1-\zeta_i)(1-\varphi_i)\gamma_i\Delta t}{\varphi_i + (1-\varphi_i) (\zeta_i (1- \rho_i \Delta t)+\rho_i\Delta t)},\\
\label{eqn:lem:DPE:G}
\frac{G_i}{1-\rho_i\Delta t}
&= -\beta_i  (1-\zeta_i)(F_i+\gamma_i\Delta t) +\zeta_i (\lambda  \beta_i - \eta).
\end{align}
Then, the function
\begin{align*}
v(M,\Delta S,Z)
&= -\frac{A_i}{2}M^2 + \frac{B_i}{2}(\Delta S)^2 - C_i M \Delta S + D_i - \frac{E_i}{2}Z^2 - F_i M Z + G_i \Delta S Z
\end{align*}
is a solution of the dynamic programming equation (DPE)
\begin{align}
\label{eqn:lem:DPE}
\frac{v(M,\Delta S,Z)}{1-\rho_i\Delta t}
&= \sup_{\Delta Z\in \RR} \lbrace f(M, \Delta S, Z;\Delta Z) + \EE[v(m,\sigma_S \sqrt{\Delta t} X,z)]\vert_{m = \beta_i \Delta S +(1-\varphi_i) M,\, z = Z+\Delta Z} \rbrace,
\end{align}
where $(M,\Delta S,Z) \in \RR^3$ and $X$ is a standard normal random variable. Moreover, the supremum on the right-hand side is attained at $\Delta Z = - \zeta_i Z$.

\item Define $\Delta \wh Z^i = (\Delta Z^i_n)_{n\geq 1}$ by $\wh Z^i_0 = Z^i_0$ and $\Delta \wh Z^i_n = -\zeta_i \wh Z^i_{n-1}$. Then, $\Delta \wh Z^i \in \mathcal{A}$ is an optimizer for \eqref{eqn:competition:objective:equilibrium:reduced} and the maximum is given by $v(M^i_0,\Delta S_1,Z^i_0)$.
\end{enumerate}
\end{lemma}

\begin{proof}
To ease the notation, we drop all sub- and superscripts ``$i$'' in the proof.

(i): After eliminating $\zeta$, the system reduces to a quadratic equation in $E$, which turns out to have a unique positive solution.

(ii): First, we show that $\Delta Z = -\zeta Z$ maximizes the supremum on the right-hand side of the DPE \eqref{eqn:lem:DPE} for our candidate value function $v$. Using that
\begin{align*}
E[v(m,\sigma_S \sqrt{\Delta t} X,z)]\vert_{z = Z + \Delta Z}
&=  -\frac{A}{2}m^2 + \frac{B}{2}\sigma_S^2\Delta t + D - \frac{E}{2}(Z+\Delta Z)^2 - Fm(Z+\Delta Z),
\end{align*}
the right-hand side of \eqref{eqn:lem:DPE} simplifies to a concave quadratic function in $\Delta Z$. Solving its first-order condition for $\Delta Z$ yields
\begin{align*}
\Delta Z
&= \frac{\eta - \beta (F+\gamma\Delta t + \lambda)}{E+\gamma\Delta t + 2\lambda}\Delta S
- \frac{(F+\gamma\Delta t)(1-\varphi)-\lambda\varphi}{E+\gamma\Delta t + 2\lambda}M - \frac{E+\gamma\Delta t}{E+\gamma\Delta t + 2\lambda} Z.
\end{align*}
Since the $Z$-coefficient is equal to $-\zeta$, it suffices to show that the other two terms vanish. We show in Lemma~\ref{lem:F} below that $F + \gamma \Delta t = \lambda\varphi/(1-\varphi)$. Hence, the $M$-coefficient vanishes and, moreover, the numerator of the $\Delta S$-coefficient simplifies to $\eta - \beta( \frac{\lambda\varphi}{1-\varphi}+\lambda ) = \eta - \frac{\beta\lambda}{1-\varphi}
= 0$ by  \eqref{eqn:thm:competition:phi} and \eqref{eqn:eta}.

Second, we show that $v$ satisfies the DPE. After inserting the optimizer $\Delta Z = -\zeta Z$ and simplifying, the right-hand side of the DPE \eqref{eqn:lem:DPE} becomes
\begin{align*}
&-\frac{1}{2}\Big[(A + \gamma \Delta t) (1-\varphi)^2\Big] M^2
+ \frac{1}{2}\Big[2\beta\eta-\beta^2(A+\gamma\Delta t)\Big](\Delta S)^2
-\Big[\beta (1-\varphi) (A+\gamma\Delta t)+\varphi \eta\Big]M \Delta S\\
&\quad+\Big[\frac{B}{2}\sigma_S^2\Delta t + D\Big]
-\frac{1}{2}\Big[(E+\gamma\Delta t)(1-\zeta)^2 + 2 \lambda \zeta^2\Big]Z^2
-\Big[(1-\zeta)(1-\varphi)(F+\gamma\Delta t) + \zeta \lambda \varphi \Big] MZ\\
&\quad+ \Big[ -\beta(1-\zeta)(F+\gamma\Delta t) + \zeta(\lambda\beta - \eta)\Big] \Delta S Z.
\end{align*}
Substituting the definitions of $A$, $B$, $C$, $D$, $F$, and $G$ from Proposition~\ref{prop:HFT} and \eqref{eqn:lem:DPE:F}--\eqref{eqn:lem:DPE:G}, one can verify that the coefficients of the terms $M^2$, $(\Delta S)^2$, $M\Delta S$, $MZ$, and $\Delta S Z$ all match those on the left-hand side of the DPE \eqref{eqn:lem:DPE}. It thus remains to consider the $Z^2$-term and show that 
\begin{align*}
-\frac{E}{2(1-\rho\Delta t)}
&=-\frac{1}{2}\Big[(E+\gamma\Delta t)(1-\zeta)^2 + 2 \lambda \zeta^2\Big].
\end{align*}
But after eliminating $\zeta$ using the second equation in \eqref{eqn:lem:DPE:E and zeta}, the right-hand side simplifies to
\begin{align*}
-\frac{E+\gamma\Delta t}{E+\gamma\Delta t + 2\lambda}\lambda
= -\lambda \zeta
= -\frac{E}{2(1-\rho\Delta t)},
\end{align*}
where we use both equations in \eqref{eqn:lem:DPE:E and zeta} for the last two equalities. We conclude that $v$ satisfies the DPE \eqref{eqn:lem:DPE}.

(iii): Let $\Delta Z \in \mathcal{A}$. We need to show that
\begin{align*}
\mathcal{J}(M_0,\Delta S_1,Z_0;\Delta Z)
&\leq v(M_0,\Delta S_1,Z_0)
\end{align*}
with equality if $\Delta Z = \Delta\wh Z$.

As $v$ satisfies the DPE, we have for each $n\geq1$ that
\begin{align*}
&(1-\rho\Delta t)^n f(M_{n-1},\Delta S_n,Z_{n-1};\Delta Z_n)\\
&\quad\leq (1-\rho \Delta t)^{n-1}v(M_{n-1},\Delta S_n,Z_{n-1}) - (1-\rho\Delta t)^n \EE_n[v(M_n,\Delta S_{n+1},Z_n)].
\end{align*}
Here, by the last assertion of (ii), the inequality turns into an equality for $\Delta Z_n =  -\zeta Z_{n-1}$. Taking the expectation $\EE_1[\cdot]$ on both sides and summing over $n=1,\ldots,N$ yields
\begin{align}
\label{eqn:lem:DPE:pf:inequality}
\begin{split}
&\EE_1\Big[\sum_{n=1}^N(1-\rho\Delta t)^n f(M_{n-1},\Delta S_n,Z_{n-1};\Delta Z_n)\Big]\\
&\quad\leq v(M_0,\Delta S_1,Z_0) - (1-\rho\Delta t)^N \EE_1[v(M_N,\Delta S_{N+1},Z_N)].
\end{split}
\end{align}

Next, we want to send $N\to\infty$ in \eqref{eqn:lem:DPE:pf:inequality}. To this end, we first observe that due to the quadratic structure of the functions $v$ and $f$, there is a constant $c>0$ such that for all $(L,\Delta S,Z,\Delta Z)\in\RR^4$,
\begin{align*}
\vert v(M,\Delta S,Z) \vert 
&\leq c(1 + M^2 + (\Delta S)^2 + Z^2),\\
\vert f(M,\Delta S, Z;\Delta Z) \vert
&\leq c (M^2+(\Delta S)^2 +Z^2 + (\Delta Z)^2).
\end{align*}
Moreover, $\Delta M, \Delta Z\in\mathcal{A}$ implies that $\sum_{n=1}^\infty (1-\rho\Delta t)^n (1+M_{n-1}^2+(\Delta S_n)^2 + Z_{n-1}^2 + (\Delta Z_n)^2)$ is integrable. This together with the estimates for $v$ and $f$ implies that the second term on the right-hand side of \eqref{eqn:lem:DPE:pf:inequality} converges to zero as $N\to\infty$ and that the left-hand side of \eqref{eqn:lem:DPE:pf:inequality} converges to $\mathcal{J}(M_0,\Delta S_1,Z_0;\Delta Z)$ as $N\to\infty$. In summary, we obtain
\begin{align*}
\mathcal{J}(M_0,\Delta S_1,Z_0;\Delta Z)
&\leq v(M_0,\Delta S_1,Z_0).
\end{align*}
Furthermore, this inequality turns into an equality for $\Delta Z = \Delta\wh Z$.
\end{proof}

The following identity is used in the proof of Lemma~\ref{lem:DPE}~(ii).

\begin{lemma}
\label{lem:F}
In the setting of Lemma~\ref{lem:DPE}, we have $F_i+\gamma_i\Delta t = \frac{\lambda\varphi_i}{1-\varphi_i}$.
\end{lemma}
\begin{proof}
This follows from equation \eqref{eqn:thm:competition:system:2} via a straightforward, but tedious calculation. We drop the subscripts as in the previous proof. Using successively the definitions of $F$ in \eqref{eqn:lem:DPE:F}, $\varphi$ in \eqref{eqn:thm:competition:phi}, and $\lambda$ in \eqref{eqn:thm:competition:lambda}, one can verify that
\begin{align*}
F +\gamma\Delta t - \frac{\lambda\varphi}{1-\varphi}
&= c \big\lbrace (\lambda  \varphi \rho -\gamma ) \Delta t (1-\varphi) + \lambda\varphi^2\big\rbrace\\
&= c \big\lbrace\beta^2 \lambda^2 (1-\rho\Delta t)-\beta  (1-\beta_\Sigma  \lambda) (\gamma  \Delta t+\lambda  (2-\rho\Delta t))+(1-\beta_\Sigma  \lambda)^2\big\rbrace\\
&= c \big\lbrace (1-\rho\Delta t) \beta^2 \beta_\Sigma^2 -(2-\rho \Delta  t + \beta_\Sigma  \gamma \Delta t)\Big(\frac{\sigma_K}{\sigma_S}\Big)^2  \beta \beta_\Sigma + \Big(\frac{\sigma_K}{\sigma_S}\Big)^4 (1-\beta \gamma \Delta t)\big\rbrace\\
&=0,
\end{align*}  
where $c$ is a nonzero term that changes from line to line.
\end{proof}

Finally, the following lemma shows that the candidate equilibrium strategies are admissible whenever the inventory management parameter $\varphi$ lies in $(0,2)$.

\begin{lemma}
\label{lem:admissible:general}
Define $M=(M_n)_{n\geq 0}$ by $\Delta M_n = \beta\Delta S_n - \varphi M_{n-1}$ for some $M_0,\beta,\varphi \in\RR$. Then $\Delta M\in \mathcal{A}$ if and only if $\varphi \in (0,2)$.
\end{lemma}

\begin{proof}
The process $M$ has the explicit representation
\begin{align*}
M_n
&=(1-\varphi)^n M_0 + \beta\sum_{j=0}^{n-1} (1-\varphi)^j\Delta S_{n-j}.
\end{align*}
Since the value increments are i.i.d.~with mean zero and variance $\sigma_S^2\Delta t$, it follows that 
\begin{align*}
\EE[M_n^2]
&=  (1-\varphi)^{2n}M_0^2 + \beta^2 \sigma_S^2 \Delta t \sum_{j=0}^{n-1} (1-\varphi)^{2j}\\
&=  (1-\varphi)^{2n}M_0^2 + \beta^2\sigma_S^2 \Delta t 
\begin{cases}
\frac{1-(1-\varphi)^{2n}}{1-(1-\varphi)^2}, & \vert 1-\varphi \vert \neq 1,\\
n, &\vert 1-\varphi \vert = 1.
\end{cases}
\end{align*}
Thus, $\EE[M_n^2]$ is bounded in $n$ if and only if $\varphi \in (0,2)$.
\end{proof}

\small
\bibliographystyle{abbrv}
\bibliography{mms}

\end{document}